\documentclass{statsoc}
\usepackage[margin=1in,footskip=0.25in, a4paper]{geometry}

\usepackage[english]{babel}
\usepackage{amsmath,amssymb, enumitem}
\usepackage[ruled]{algorithm2e}
\usepackage{bbm, dsfont}
\usepackage{etoolbox}

\makeatletter
\patchcmd{\@makecaption}
  {\parbox}
  {\advance\@tempdima-\fontdimen2} 
  {}{}
\makeatother  
\usepackage{xcolor}
\usepackage{graphicx}
\usepackage[caption = false]{subfig} 
\usepackage{subfig}
\usepackage{booktabs}

\usepackage{natbib}
\usepackage{hyperref}
\usepackage{graphicx}
\usepackage{ntheorem}

\usepackage[perpage]{footmisc}

\newtheorem{theorem}{Theorem}
\newtheorem{proposition}{Proposition}
\newtheorem{lemma}{Lemma}

\newtheorem{corollary}{Corollary}
\newtheorem{definition}{Definition}
\newtheorem{example}{Example}

\newtheorem{remark}{Remark}
\newtheorem{assumption}{Assumption}

\theoremstyle{break}

\makeatletter
\@ifpackageloaded{hyperref}{%
  \def\elabel#1{\label{#1}}%
}{%
  \def\elabel#1{\@bsphack \protected@write \@auxout {}{\string \newlabel {#1}{{\theAlgoLine }{\thepage }}}\@esphack}%
}
\makeatother
\def\g{\mathbf}

\DeclareMathOperator{\argmin}{arg\,min}

\usepackage[utf8]{inputenc} 

\newcommand{\eps}{\varepsilon}

\newcommand{\einsfun}{\mathbbm 1}

\usepackage{pgfplots}
\pgfplotsset{compat=newest}

\newcommand{\Ptr}{\mathbb P}
\newcommand{\Pte}{\mathbb Q}

\newcommand{\Qtesttau}{\mathbb{Q}_{\tauzero}}
\renewcommand{\P}{\mathbb P}

\newcommand{\bfX}{\mathbf X}

\newcommand{\tauzero}{\tau_0}
\newcommand{\tzero}{t_0}

\def\M{\text{median}}

\newcommand{\beao}{\begin{eqnarray*}}
	\newcommand{\eeao}{\end{eqnarray*}\noindent}

\newcommand{\beam}{\begin{eqnarray}}
	\newcommand{\eeam}{\end{eqnarray}\noindent}

\title[Progression: an extrapolation principle for regression]{Progression: an extrapolation principle for regression}

\author{Gloria Buritic\'a}
\address{
    Université Paris-Saclay, AgroParisTech, INRAE, 
	UMR MIA Paris-Saclay, 22 place de l'Agronomie,
	91123 Palaiseau, France
  }

\author[G. Buritic\'a and S.~Engelke]{Sebastian Engelke}
\address{
	Research Institute for Statistics and Information Science, University of Geneva, Boulevard du Pont d'Arve 40, 1205 Geneva, Switzerland.
}

\begin{document}

\begin{abstract}
	The problem of regression extrapolation, or out-of-distribution generalization, arises when predictions are required at test points outside the range of the training data. In such cases, the non-parametric guarantees for regression methods from both statistics and machine learning typically fail. Based on the theory of tail dependence, we propose a novel statistical extrapolation principle. After a suitable, data-adaptive marginal transformation, it assumes a simple relationship between predictors and the response at the boundary of the training predictor samples.	
	This assumption holds for a wide range of models, including non-parametric regression functions with additive noise. Our semi-parametric method, \texttt{progression}, leverages this extrapolation principle and offers guarantees on the approximation error beyond the training data range.
	We demonstrate how this principle can be effectively integrated with existing approaches, such as random forests and additive models, to improve extrapolation performance on out-of-distribution samples.
\end{abstract}


\section{Introduction}

Regression methods from both statistics and machine learning model the relationship between a response variable $Y$ and a set of predictors $\mathbf X$ in a $p$-dimensional space. These methods aim to predict $Y$ given a new predictor value $\g x\in \mathbb R^p$, typically through the conditional mean, or, as is the focus of this paper, the conditional median
\begin{align}\label{cond_median}
	 m(\g x) := \M(Y \mid \g X = \g x). 
\end{align}
Flexible non-parametric methods such as random forests or neural networks perform well at interpolation, that is, prediction within the range of the training data. 
However, classical non-parametric guarantees break down for extrapolation, when the test point $\g x$ lies outside the domain of the training data \citep{xu2021how}. While extrapolation in parametric models is mathematically straightforward, it comes with the drawback of requiring strong structural assumptions about the relationship between the predictors and the response. Such assumptions are particularly hard to justify in regions of the predictor space where training data is sparse or absent.

In climate science, where impacts $Y$ are typically modeled as functions of environmental variables $\g X$ (e.g., temperature, precipitation), models are trained on current climate data
and applied to future climate projections, where the distribution of predictors has shifted in mean and/or variance due to climate change. 
Although the challenge of extrapolation is frequently acknowledged in this field, no universally effective solution exists. For example, in insurance \citep{pim2021} or fire risk modeling \citep{hei2023} based on climate drivers, extrapolation is often assumed constant above the largest training observation, or by excluding test points outside the training range.
In machine learning, extrapolation is known as the out-of-distribution generalization problem, encountered in settings like image recognition \citep{gan2016}, protein fitness prediction \citep{fre2024}, and large language models \citep{graphcast}.
In weather forecasting, artificial intelligence methods such as graph neural networks \citep{graphcast} outperform physical models on average but show limitations in predicting extreme events, particularly when it involves extrapolation \citep{pasche2024}.

Mathematically, extrapolation occurs if the test point $\g x$ lies far from the training samples or outside the training distribution $\Ptr$, based on some metric.
One common cause is covariate shift, where the test distribution $\Pte$ differs from the training distribution, increasing the likelihood of encountering test points $\g x$ outside the training support.
Even when the test and training distributions are identical, the stochastic nature of sampling may result in test points that are more extreme than any observation in the training set. Extrapolation is also required when the test point $\g x$ is deterministically set to a value of interest beyond the range of the training data.

Extrapolation is inherently challenging, and any method offering statistical guarantees must impose additional assumptions on the structure or regularity of the data-generating process; see Section~\ref{sec:literature} for a brief literature review. Methods that assess point-wise performance at a new test point $\g x$ often assume that the point is outside a fixed, bounded support of $\g X$ under $\Ptr$ \citep{shen:meinshausen:2023,pfister:buhlman:2023}. Theoretical guarantees for extrapolation beyond this support are typically provided on population level or rely on the assumption that the sample size on the fixed support grows to infinity.

We consider a related but fundamentally different perspective on extrapolation.
In many applications, in particular in environmental and climate sciences, extrapolation is inherently linked to the sample size $n$.
As such, out-of-distribution generalization should be defined relative to the training sample, rather than the population version $\Ptr$.
Assume $n$ samples $X_1,\dots, X_n$ from a univariate predictor, with $X_{(1)}$ and $X_{(n)}$ denoting the smallest and largest observations, respectively. Non-parametric methods typically fail to predict at test points $x$ outside of $[X_{(1)}, X_{(n)}]$, and may even struggle near the boundary of this range. As the sample size increases, the range of observed data expands, necessitating extrapolation beyond the new boundaries. 

To capture this mathematically, we propose a semi-parametric approximation $\tilde m_{\tauzero}$ of the conditional median function $m$ in~\eqref{cond_median}, which can be estimated on the training domain $[Q_X(1-\tauzero), Q_X(\tauzero)]$, where $\tauzero$ is close to one. In practice, $\tauzero = 1-k/n$ depends on the sample size $n$ and includes a tuning parameter $k<n$. Our goal is to provide extrapolation guarantees on significantly larger domains $[\nu^-(\tauzero), \nu^+(\tauzero)]$, where the extrapolation limits $\nu^\pm(\tauzero)$ depend on the distributional properties of the model. This sample-based perspective is well-known in extreme quantile regression \citep[e.g.,][]{che2005,gne2024}, which addresses extrapolation beyond the training range of the response $Y$, rather than in the predictor space.

The rationale behind our principled approach to extrapolation is inspired by the work of \cite{box1964}. Their power transformations of the response variable are designed to simplify the relation with predictors, allowing a simple model (e.g., linear) to fit well. 
We consider a different class of parametric transformations, justified by asymptotic theory, which can be estimated from data. 
Crucially, we transform both predictor and response variables into $X^*$ and $Y^*$ with a common marginal distribution, and we do not assume that a parametric model holds on the whole predictor space. Instead, our key assumption is that the functional dependence on the transformed scale simplifies at the boundary of the training data, where the 
conditional median follows a leading linear term
\begin{align}\label{main_assumption_intro}
	\M(Y^* \mid X^* = x^*) = a x^*  +  (x^*)^\beta b  + r(x^*), \quad  x^* \to \infty,
\end{align}
where $a\in [-1,1]$, $b\in \mathbb R$, $\beta \in [0,1)$ and the remainder term is negligible with $r(x^*) = o(1)$ as $x^*\to\infty$.
This surprisingly simple relationship in the tail regions is motivated by the theory of dependence between extreme values in multivariate data \citep{resnick:2007, heffernan:tawn:2004}.   
We argue that these ideas are highly relevant in the framework of regression extrapolation.
In fact, for a broad class of additive noise models $Y = f(X) + \varepsilon$, as well as for the pre-additive noise model $Y = g(X + \eta)$ considered in \texttt{engression} \cite{shen:meinshausen:2023}, we demonstrate that this assumption holds.

Our new methodology leverages this principle of regression extrapolation, abbreviated as the \texttt{progression} method. The name also reflects a fundamental departure from the original meaning of regression in \cite{gal1886}, who observed that extreme samples in hereditary stature data tend to regress toward the overall mean. In contrast, methods based on \texttt{progression} can extrapolate towards values more extreme than in the training data. 
Figure~\ref{fig:intro} illustrates our principled approach to extrapolation. While the data may exhibit arbitrarily complex behavior on the original scale (left panel), it simplifies to an approximately linear relationship at the boundary of the training predictor distribution on the transformed scale (right panel). 
Our \texttt{progression} method fits the conditional median function $m$ (orange line) by first learning appropriate monotone transformations of both the predictor and response variables, and then fitting the parametric approximation from~\eqref{main_assumption_intro} on the transformed scale (blue line). By inverting the transformation, we obtain an accurate semi-parametric approximation $\tilde m_{\tauzero}$ of the true function $m$.
For any test distribution $\mathbb Q_{\tauzero}$ supported on the larger interval $[\nu^-(\tauzero), \nu^+(\tauzero)]$, we obtain the extrapolation guarantee
\begin{align}
	\mathbb E_{\mathbb Q_{\tauzero}} \left( \frac{\tilde m_{\tauzero}(X)- m(X)}{m(X)} \right) \to 0, \qquad \tauzero \to 1.
\end{align}

Building on this principle, we propose the random forest \texttt{progression} method for univariate predictors $X$. This method uses the localizing weights from a regression random forest \cite{breiman:2001} to obtain a smoother of $Y$ as a function of $X$ within the sample range, while applying an appropriate parametric approximation in the boundary regions of the training data.		
This allows for a smooth, data-driven transition from a non-parametric random forest fit to the extrapolation-aware $\texttt{progression}$ method inspired by~\eqref{main_assumption_intro}.  
Our extrapolation principle can be extended in various ways to handle $p$-dimensional predictors $\g X$. Here, we adopt the framework of additive models. By leveraging random forest-based extrapolation for individual predictors, we propose an adaptation of the backfitting algorithm \citep{has1987} to implement the additive model \texttt{progression}. This approach provides accurate extrapolation for multiple predictors under an additive structure.

In Section~\ref{sec:literature}, we give a brief overview of different perspectives on extrapolation and the related literature. Section~\ref{sec:prelim} covers the extrapolation of marginal distributions, which is crucial for the transformations to $X^*$ and $Y^*$ in~\eqref{main_assumption_intro}. Our extrapolation principle for regression, the foundation of the \texttt{progression} methods, is presented in Section~\ref{sec:extrapolation_principle}, where we also establish the main result on the uniform bound for the relative approximation error outside the training distribution. In Section~\ref{sec:statistics}, we explore statistical aspects and introduce the random forest and additive model \texttt{progression} methods. A simulation study in Section~\ref{sec:experiments} demonstrates the improved performance  of \texttt{progression} under domain shifts compared to classical methods such random forest or local linear forest \citep{friedberg:tibshirani:wager:2021}, and also compares it to the recent extrapolation method \texttt{engression} \citep{shen:meinshausen:2023}. 

All proofs are provided in the appendix.

\begin{figure}[h!]
	\begin{center}
		\includegraphics[width=.35\linewidth]{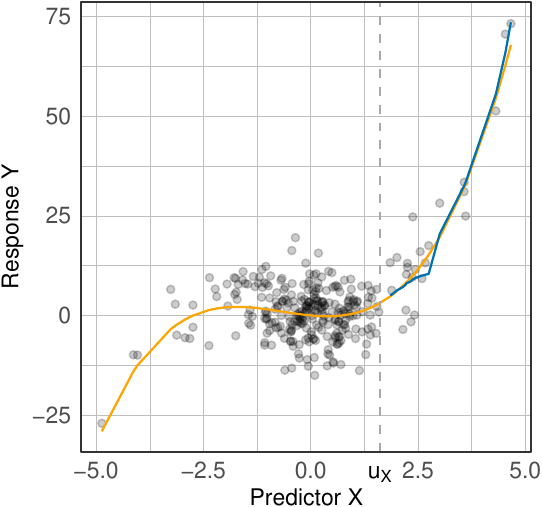}%
		\includegraphics[width=.335\linewidth]{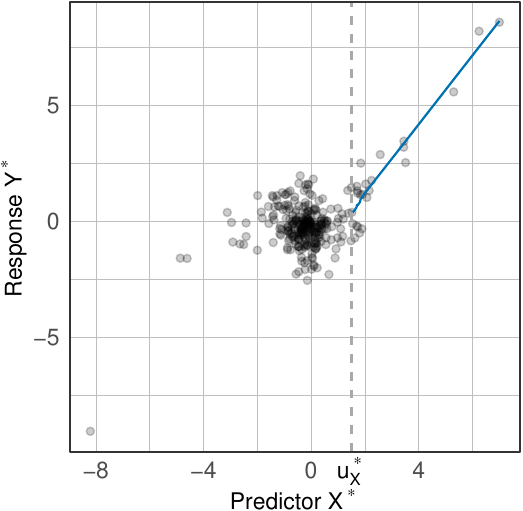}
	\end{center}
	\caption{Left: training data (grey points) of $(X,Y)$ sampled from model $Y = f(X) +\varepsilon$ with regression function $f$ (orange line) and approximation from the extrapolation principle (blue line). Right: transformed training data on Laplace scale $(X^*, Y^*)$ and linear approximation above threshold $u_{X}^*$ (blue line).
	}
	\label{fig:intro}
\end{figure}

\subsection{Related work}\label{sec:literature}
We study the problem of extrapolation in non-parametric regression beyond the range of the training data.  
Extrapolation, or of out-of-domain generalization, is difficult and has many facets. There is not one method that is superior in all applications, since the different assumptions will only be satisfied by certain data structures.
Our approach makes generic assumptions on the training distributions of the noise and predictor variables, rather than explicit structural assumptions on the model. It relies on tools of extreme value theory that allow for accurate approximations of dependence between predictors and response at the boundary of observed predictor samples. The sample-based notion of extrapolation is different from most existing methods that rely on population quantities.

\cite{shimodaira:2000} first studied covariate shifts and introduced 
importance-reweighted maximum likelihood estimators. The weights 
are likelihood ratios between the training and test distribution. 
Subsequent work studied these methods in the context of cross validation \citep{sugiyama:krauledat:muller:2007} and derived theoretical guarantees \cite[e.g.][]{cortes:mansour:mohri:2010,ma:pathak:wainwright:2023}. 
Using likelihood ratios, which are often assumed to be bounded, requires the test distribution to be dominated by the training distribution; for a reference book see \cite{quinonero:sugiyama:schwaighofer:lawrence:2022}. This excludes the out-of-distribution extrapolation scenarios studied in this article.

Another line of research investigates distributionally robust optimization methods that minimize the worst-case risk over a neighborhood of around the training data. This requires assumptions on the distances between adversarial distributions and the training sample, and several divergences and metrics have been considered \citep{kuh2019, NIPS2016_4588e674}. In contrast, our method based on principles of extreme value theory relies on an asymptotically motivated distance measure; see also \cite{vig2020} where such a distance is used in novelty detection.

There is an extensive body of literature on structural robustness, where distributional changes can be explained through interventions in a causal system
\citep[e.g.,][]{rothensausler:meinshausen:buhlmann:peters:2021,buhlmann:2020,christiansen:pfister:jakobsen:gnecco:peters:2021,gnecco:peters:engelke:pfister:2023}. The extrapolation is however typically linear, and moreover domain knowledge of the causal graph is required. 

Non-linear extrapolation beyond the observed data-range has been studied in  \citep{dong:ma:2023} who show that it is necessary to restrict the non-linear function class to have control over the extrapolation guarantees. 
\cite{pfister:buhlman:2023} assume bounds on the directional derivatives of the regression function on the training domain to derive extrapolation-aware bounds outside of this domain. The \texttt{engression} method of \cite{shen:meinshausen:2023} proposes a methodology for regression outside of the training domain with theoretical guarantees for univariate predictors in a pre-additive noise model. As described in the introduction, our approach differs in the sample-based perspective on extrapolation, and the fact that we obtain theoretical guarantees also for the widely used additive noise models.

Our semi-parametric, non-linear extrapolation method relies on principles from extreme value theory. Traditionally, this field studies risk assessment and uses extrapolation in the direction of the response variable $Y$ in terms of high quantiles. If predictors are present, there is substantial research on flexible models for extreme quantile regression \citep{ExGAM2,Kohfireboosting, gbex, gne2024, Pasche2025}.
Some works have considered classification \cite{jalalzai:clemencon:sabourin:2018} and regression in extreme regions \cite{huet2024regressionextremeregions}, under certain distributional assumptions and bounded regression functions.
\cite{bou2022} use univariate extreme value extrapolation in combination with generative adversarial networks to generate climate data beyond historical records.
In this paper we argue that principles from multivariate extreme value theory can be leveraged for regression extrapolation in the predictor space, and our \texttt{progression} method formalizes this idea.

\section{Preliminaries: marginal tail extrapolation}\label{sec:prelim}

\subsection{First- and second-order conditions for tail extrapolation}
\label{sec:first_order}
Standardizing the margins of a random variable $Y$ requires knowledge of its distribution function, 
which in practice is typically unknown and needs to be estimated. 
In our case, we are interested in the tail of the distribution and for this reason we rely on classical extreme value theory \citep[][]{dehaan:ferreira:2006}. 

Let $F$ be the distribution function of a random variable $Y$ with upper endpoint $u^* \in (-\infty, \infty]$.
We say that $F$ is in the domain of attraction of a generalized Pareto distribution 
with shape parameter $\gamma \in \mathbb{R}$ 
if there exists a positive scale function $\sigma :\mathbb{R}\to (0,\infty)$ such that 
\begin{equation}
    \label{eq:domain:attraction:1}
\lim_{u \to u^*} \P(Y - u \leq x \sigma(u) \mid  Y > u)  = 1- (1+\gamma x)_+^{-1/\gamma} := H_\gamma(x), \quad  x \geq 0,
\end{equation}
where $H_\gamma$ is the distribution function of the generalized Pareto distribution and $a_+ = \max(0,a)$ for $a\in\mathbb R$; for $\gamma = 0$ it should be understood as the limit as $\gamma \downarrow 0$, that is, $H_0(x) = 1-\exp(-x)$ 
\cite[see, e.g.,][]{balkema:dehaan:1974,pickands:1975}.
The shape parameter determines the heaviness of the tail of $Y$, ranging from heavy tails for $\gamma > 0$, light tails for $\gamma = 0$, and distributions with finite upper endpoint $-1/\gamma$ for $\gamma < 0$.
Note that when $\gamma = 0$, there are certain distributions in the domain of attraction of $H_0$ that have a finite upper endpoint $u^* < \infty$. 
In this paper, we only use results where $Y$ has an infinite upper endpoint $u^* = \infty$, and focus on this case in the sequel. 
The shape parameter is then always non-negative $\gamma \geq 0$.

The importance of the limit~\eqref{eq:domain:attraction:1} is that it
yields a strategy to approximate tail probabilities and quantiles of $F$.
Indeed, for a high threshold {$u \in \mathbb{R}$,} we define the approximate distribution function based on the generalized Pareto distribution (GPD) as
\beam \label{eq:tail:approximation:tilde:1}
\tilde F(x) = 1 - \{1-F(u)\}\left\{1- H_\gamma\left(\frac{x-u}{\sigma(u)}\right)\right\},\quad x > u,
\eeam 
and $\tilde F(x) = F(x)$ if $x \leq u$.
In addition, by inverting this approximation, we obtain an approximation of the quantile function above level $\tauzero = F(u)$ as 
\beam \label{eq:quantile:tail:approximation}
\tilde Q(\tau) = u  + \frac{\sigma(u)}{\gamma} \left\{ \left(\frac{1-\tauzero}{1-\tau}\right)^{\gamma} - 1 \right\}, \quad \tau > \tauzero,
\eeam 
and $\tilde Q(\tau) = Q(\tau)$ for $\tau \leq \tauzero$. 
In order to compute these extrapolations in a statistical application,
 we replace the intermediate quantile $ u= Q(\tauzero)$ and the GPD parameters $\sigma = \sigma(u)$ and $\gamma$ with suitable estimators. The approximations above then allow us to estimate tail probabilities and quantiles beyond the data range. Numerous practically relevant distributions satisfy the first-order conditions that we introduce here; see Section~\ref{subsec:examples:tails}. 
  The statistical aspects will be discussed in Section~\ref{sec:stat_univariate}.

\begin{remark}\label{rem:sigma_U}
	In the case $\gamma=0$ the scale function satisfies
	$\lim_{u \to \infty} \sigma(u)/u = 0;$
	see Appendix~\ref{sec:app:1} and \citet[Lemma 1.2.9]{dehaan:ferreira:2006}.
	A weak additional assumption needed later is that there exists $\delta > 0$ such that
	\begin{align}\label{ass_aU}
		\lim_{u \to \infty}	\frac{\sigma(u) }{u}
	   \left\{- \log \bar F(u) \right\}^{\delta} = 0.
   \end{align}	
	This assumption holds for almost all practically relevant distributions. 
\end{remark}

In order to derive sharp statistical guarantees and to obtain uniformly valid approximations, it is customary to make a second-order assumption that quantifies the rate of convergence in~\eqref{eq:domain:attraction:1}. 
We say that a distribution $F$ satisfies the second-order domain of attraction condition if there exists $\Sigma : \mathbb{R} \to \mathbb{R}$, not changing sign eventually, such that $\lim_{u \to \infty} \Sigma(u) = 0$, and a second-order parameter $\rho \leq 0$, such that there exists a non-zero function $h_{\gamma,\rho}$ with
\beam \label{eq:second:order:intro}
\lim_{u \to \infty} \frac{1}{\Sigma(u) }\left\{ \frac{ \bar F( u + x\sigma(u)) }{ \bar F (u)} -  \bar H_{\gamma}(x)\right\} = h_{\gamma,\rho}(x), \quad x \geq 0,
\eeam 
where the explicit form of possible limit functions $h_{\gamma,\rho}$ is given in~\eqref{eq:hrho}.
Intuitively, this second-order condition quantifies the bias in the approximation~\eqref{eq:domain:attraction:1}, where  smaller values of $\rho$ lead to better rates of convergence; in particular $\Sigma$ is regularly varying with index $\rho \gamma$.
We refer to Appendix~\ref{sec:appendix} for a review on domain of attraction conditions.

A direct consequence of the domain of attraction conditions is that the approximation error vanishes in the tail, both for extrapolation 
of probabilities and quantiles. The second-order condition quantifies how far extrapolation is possible.

\begin{lemma}\label{lem:uniform:tail:approximation}
	Let $Y$ be a random variable with distribution function $F$ with infinite upper endpoint that satisfies the first-order domain of attraction condition with shape parameter $\gamma \geq 0$ and scale function $\sigma$.
	Then, for any sequences $\theta(\tau_0) $ and $\nu(u) $ such that
		\begin{align}\label{eq:extra:condition}
			\lim_{u \to \infty} \left\{ \frac{\bar F(u)}{\bar F(\nu(u))} \right\} = \lim_{\tau_0 \to 1}  \left\{ \frac{1-\tau_0}{1-\theta(\tau_0)} \right\} < \infty,
		\end{align}
		the tail and quantile tail extrapolations in~\eqref{eq:tail:approximation:tilde:1} and~\eqref{eq:quantile:tail:approximation} satisfy
	\begin{align} \label{eq:f:tilde:approx} 
		&\lim_{u \to \infty} 
		\sup_{ { u <   x < \nu(u) } }
		\Big| \frac{1-\tilde F(x)}{1-F(x)} - 1 \Big| = 0, \\
 		& \lim_{\tauzero  \to 1} 
		\sup_{ \tauzero < \tau < \theta(\tau_0)}   \Big|\frac{ \tilde Q(\tau)}{  Q(\tau) }  - 1\Big| = 0,    \label{eq:U:tilde:approx}
		\end{align}
		respectively.
	In addition, if $Y$ satisfies the second-order domain of attraction condition with auxiliary function $\Sigma$ and second-order parameter $\rho \leq  0$, 
	then \eqref{eq:f:tilde:approx} and \eqref{eq:U:tilde:approx} also hold for all sequences $\nu(u)$ and $\theta(\tau_0)$ such that 
	 there exists $\eta > 0$ with 
	\begin{align}\label{eq:extra:condition:rho1}
		\lim_{u \to \infty} \Sigma(u)\left\{ \frac{\bar F(u)}{\bar F(\nu(u))} \right\}^{\rho + \eta} = \lim_{\tau_0 \to 1} \Sigma\{Q(\tau_0)\} \left\{ \frac{1-\tau_0}{1-\theta(\tau_0)} \right\}^{\rho + \eta} = 0.
	\end{align}
	In particular, if $\rho<0$ we may always choose $\nu(u) =\infty$ and $\theta(\tauzero) = 1$.
\end{lemma}

\begin{remark}\label{rem:lower_tail}
	In this section we only discussed the upper tail of $Y$. The lower tail can be treated in an analogous way by replacing $Y$ by $-Y$, yielding potentially different first- and second-order tail parameters $\gamma^-$ and $\rho^-$, respectively.
\end{remark}

\subsection{Distributions satisfying first- and second-order conditions}\label{subsec:examples:tails}

While first- and second-order domain of attraction conditions might seem technical, it is important to note that they are satisfied for most of
the commonly used distributions. In this section we discuss a few classes of examples for later reference, 
and we show that all functions and coefficients can be computed explicitly; for a more detailed review we refer to \cite{fraga:gomes:dehaan:neves:2007}.
We start with classes of exponential and Pareto distributions, which correspond to the case of no approximation bias in~\eqref{eq:domain:attraction:1} 
for $\gamma =0$ and $\gamma >0$, respectively.

\begin{example}[Exponential distribution]\label{ex:exponential:distribution}
	The exponential distribution $F(x)  = 1 - e^{-\lambda x}$ for $x >0$ and $\lambda > 0$, satisfies first- and second-order conditions for $\sigma(u) \equiv 1$, $\gamma = 0$, and $\rho = - \infty$. Lemma~\ref{lem:uniform:tail:approximation} then holds with $\nu(u) = \infty$.
\end{example}

\begin{example}[Pareto distribution]
	The Pareto$(\alpha)$ distribution $F(x) = 1 - x^{-\alpha}$ for $x > 1$ and $\alpha > 0$, satisfy first- and second-order conditions for $\sigma(u) = u/\gamma $, $\gamma = 1/\alpha$, and $\rho = - \infty$.
	Lemma~\ref{lem:uniform:tail:approximation} then holds with $\nu(u) = \infty$.
\end{example}

The interpretation of a second-order parameter $\rho = -\infty$ is that the
right-hand side of~\eqref{eq:second:order:intro} is zero for all $x\geq  0$.
This is an idealized case that typically does not occur in practice. The normal distribution, for instance, has a very slow convergence rate. 

\begin{example}[Normal distribution]
	The standard normal distribution satisfies 
	\[F(x) = 1 - e^{-x^2/2}\{x^{-1} - x^{-3} + o(x^{-3})\}/\sqrt{2\pi }, \quad x\to \infty.\] 
	One can show that 
	$\sigma(u) = 1/u$, $\Sigma(u) = O(1/u^2)$, 
	and $\gamma = \rho = 0$;
	 see Example 1.1.7 and Exercise 2.9 in \citet{dehaan:ferreira:2006}. 
	With the choice $\nu(u) = O(u^{1+\kappa})$, 
	for any fixed  $\kappa > 0$, Lemma~\ref{lem:uniform:tail:approximation} holds. 
\end{example}
We observe that in this case, uniform extrapolation of tail probabilities and quantiles is only possible until a level that is determined by the second-order behavior. A similar result holds for the larger class of Weibull distributions.

\begin{example}[Weibull distribution]\label{example:weibull}
	The family of Weibull distributions $F(x) = 1 - e^{-\lambda x^\tau},$ $x\geq 0$, with $\lambda, \tau > 0$, satisfies first- and second-order conditions with 
	$\sigma(u) = 1/(\lambda \tau u^{\tau - 1})$, $\Sigma(u) = O(1/u^\tau)$, and $\gamma= \rho =0$.
	Lemma~\ref{lem:uniform:tail:approximation} holds with $\nu(u) = O(u^{1+\kappa})$, for any $\kappa > 0$.
\end{example}

\begin{example}[Hall-type distribution]\label{ex:hall}
	A well-known class of heavy-tailed distributions with $\gamma > 0$ are
	Hall-type distributions defined by their expansion of the tail quantile function $U(t) = Q(1-1/t)$, $t\geq 1$ as 
	\begin{equation}
		U(t) = c t^{\gamma}\left( 1 + c_1t^{-\beta_1} + c_2 t^{-\beta_2} + \dots \right),
	\end{equation}
	where $c, \gamma >0$, $c_i \geq 0$, $i \in \mathbb N$, and $0 \leq \beta_1 < \beta_2 < \dots$. 
	All Hall-type distributions satisfy the second-order condition \citep[][Section 3.3.2]{beirlant:goegebeur:teugels:segers:dewaal:ferro:2004}.
	A member of this class is the Burr distribution $F(x) = 1 - \eta^\lambda (\eta + x^\tau)^{-\lambda}$ for $\lambda, \tau , \eta > 0$ and $x \geq 0$. It has first- and second-order parameters $\gamma = \tau \lambda$ and $\rho = -1/\lambda < 0$, respectively, and
	 $\sigma(u) = \eta \left( 1 + u^\tau/\eta\right)^{1/\lambda} /(\lambda \tau u^{(\tau -1)}) $, $\Sigma(u) = O(1/u^\tau)$. Lemma~\ref{lem:uniform:tail:approximation} holds with $\nu(u) = \infty$.
\end{example}

\section{Regression extrapolation principle}\label{sec:extrapolation_principle}

\subsection{Modeling assumption}\label{sec:model_assumption}

Regression modeling studies the influence of a predictor variable $X$ on the target response variable $Y$. We derive in this section a general extrapolation principle for the median regression function
\beam \label{eq:ocnd:median} 
	m(x) = \M(Y \mid X = x), \qquad x\in\mathbb R,
\eeam 
that enables the prediction beyond the range of observed predictors.
Throughout, we assume that during training, we observe the pair $(X,Y)$ under a given training distribution $\Ptr$. Our approach does not require knowledge of the test distribution. Instead, we first focus on the point-wise error at a new, deterministic test point $x\in\mathbb R$. We adopt the standard assumption from the covariate shift literature that the conditional distribution $Y\mid X=x$ remains the same at both training and test time.
This assumption allows us to infer the behavior of $m(x)$ at test points based on the observed dependence between $X$ and $Y$ in the training data.  
Our primary interest lies in extrapolation, where the test point $x$ falls outside the range of the training data, that is, as $x \to \pm \infty$. Consequently, we must analyze the relationship between predictors and the response in the tail regions, which are sparsely represented in the training set.

We assume $X$ and $Y$ admit continuous distribution functions $F_X$ and $F_Y$, respectively, and we study the random variables $X^*$ and $Y^*$ transformed to have Laplace margins with distribution function $F_L(x) = \exp(x)/2$ for $x<0$ and $F_L(x) = 1 - \exp(-x)/2$ for $x\geq 0$. 
  We obtain the standardized variables through the probability integral transform,
\begin{align} \label{eq:laplace:trans}
X^* = Q_L \circ F_X(X), \qquad Y^* = Q_L \circ F_Y(Y),
\end{align}
where $Q_L$ is the quantile function of the Laplace
distribution $F_L$. Since monotone functions can be pulled out of the median, the regression function can then be rewritten as
\begin{align}\label{f_rewritten} 
	m(x) 
	&= Q_Y \circ F_L\{ \text{median}(Y^* \mid X^* = x^* ) \}, 
\end{align}
where $x^*= Q_L \circ F_X(x)$ is the transformation of the point of interest~$x$. Note that the transformation $Q_Y \circ F_L$ takes the data from Laplace scale to the original scale of the response.
Instead of regressing the median on the original scale, we may therefore fit a regression model on the transformed scale $(X^*, Y^*)$. 
The fitted median function can then be transformed back to the original scale to obtain a prediction of the original median function $m$.

Our approach leverages the fact that 
by standarizing margins, the extremal relation between the transformed variables $X^*$ and $Y^*$ simplifies, even for complex regression functions $m$. 
To illustrate this, consider the additive noise model $Y= f(X) + \varepsilon$ with $\M(\varepsilon) =0$, and note that $m=f$.
Inverting formula~\eqref{f_rewritten} we obtain for the median function on transformed scale
\beao 
  m^*(x^*) := \M( Y^* \mid X^* = x^*) =  Q_L \circ F_Y \circ f\{ Q_X \circ F_L (x^*)\}.
\eeao 
If the model is noiseless and $f$ is a strictly monotone function, we can compute straightforwardly the distribution of $Y$ as $F_Y = F_X \circ f^{-1}$ and plug it in in the above expression. A direct calculation then yields $ m^*(x^*) = x^*$. 
We generalize this argument in Section~\ref{sub:sec:examples} and show that when the extremal dependence between $X$ and $Y$ is strong, then the conditional median of standardized Laplace margins is well approximated by a first-order linear term as $x^* \to \infty$. 
This motivates the following assumption.

\begin{assumption}\label{main_assumption_text}
	Let $X$ and $Y$ be random variables with $X^*$ and $Y^*$ as defined in~\eqref{eq:laplace:trans}.	
	Assume there exists $a\in[-1,1]$, $\beta \in [0,1)$, $b \in \mathbb R$,	with $\beta=0$ if $|a|=1$, and $\beta, b \neq0$ if $a=0$, 
	and a rate function $r: \mathbb{R} \to \mathbb{R}$ such that the approximation of the median on transformed scale behaves as
    \begin{align}\label{convergence:median:1}
		| \M( Y^* \mid X^* = x^*) - a x^*  -  (x^*)^\beta b  |  = r(x^*),
	\end{align}
	and the rate function $r^*$ satisfies $r(x^*) = o(1)$ as $x^*\to\infty$.
\end{assumption}

In the case where $Y$ satisfies a domain of attraction conditions~\eqref{eq:domain:attraction:1} with $\gamma_Y = 0$, we will often relax this assumption to $r(x^*) = O\{\log(x^*)\}$. For positive shape parameter $\gamma_Y>0$ we always require that the approximation converges to the true conditional median. 
The case $|a| = 1$ corresponds to strong dependence between $X$ and $Y$ for large values of the predictor.
More generally, if $a > 0$, there is positive tail association between $X$ and $Y$ and we have $m^*(x^*) \to \infty$ as $x^* \to \infty$. On the other hand, for $a< 0$, there is negative association and $m^*(x^*)  \to - \infty$, as $x^* \to \infty$.
For $a=0$, tail dependence between response and predictor is weak and the conditional median on transformed scale is quantified by the sub-linear term with parameter $\beta$. 
Assumption~\ref{main_assumption_text} is inspired by multivariate regular variation in extreme value theory for $|a|=1$ \citep{resnick:2007}, and by the work of \cite{heffernan:tawn:2004} for general $a\in[-1,1]$ who used a similar assumption in dependence modelling of copulas.

While at first sight, Assumption~\ref{main_assumption_text} seems very strong, it is indeed surprisingly mild. Indeed, in Section~\ref{sub:sec:examples}, we show that it is satisfied by a wide range of  regression models, including non-linear additive and pre-additive noise models.

\subsection{Extrapolation principle}
\label{sec:extrapolation:principle}

Representation~\eqref{f_rewritten} and Assumption~\ref{main_assumption_text} 
motivate an extrapolation principle for the conditional median for large values $x$ of the predictor space. 
We define the \texttt{progression} approximation above the quantile at probability level $\tau_0\in(0,1)$ as
\begin{align}\label{extr_pinciple}
	\tilde m_{\tauzero}(x) =
	\tilde Q_Y \circ F_L\{ a \tilde x^* + (\tilde x^*)^\beta b \}, \quad x > Q_X(\tau_0),
\end{align}
 where $\tilde x^* = Q_L \circ \tilde F_X(x)$ and we typically omit the dependence on $\tau_0$ in $\tilde m = \tilde m_{\tau_0}$ for notational simplicity.
 The functions $\tilde F_X$ and $\tilde Q_Y$ denote the GPD approximations of the distribution function of $X$ and quantile function of $Y$, respectively; see Section~\ref{sec:first_order}.
Figure~\ref{fig:intro} in the introduction illustrates this extrapolation principle. The right panel shows simulated data (grey points) from and additive noise model $Y = f(X) + \varepsilon$ for a particular regression function $f$ (orange line). In regions with few or no training data such as in the upper tail of the predictor $X$, predictions of non-parametric or machine learning methods would be very poor. 
The left panel shows the same data transformed to the Laplace scale.
In the bulk of the training data, no particular structure is apparent, but for predictor values $x^*$ larger than a threshold $u_{X}^* =  Q_L(\tauzero)$, the regression on this scale in fact seems to be linear (blue line). Transforming this approximation back to the original scale yields an accurate extrapolation of the (highly non-linear) regression function $f$ above the corresponding threshold $u_X = Q_X(\tauzero)$.
We discuss here the case where the predictor value $x$ of interest is out-of-distribution since it is larger than most or all training samples, that is, $x > u_X$ for a high in-sample level.
In practice, the level $u_X =  Q_X(\tauzero)$ is typically chosen relative to the training sample size $n$ by setting $\tauzero = 1-k/n$ for some tuning parameter $k<n$; see Section~\ref{sec:stat_univariate}
for statistical details. 
Extrapolation at the lower end of the training data where $x < l_X = Q_X(1-\tauzero)$ is done analogously.

Our extrapolation principle for the conditional median in~\eqref{extr_pinciple} involves three layers of approximations of the target quantities in~\eqref{f_rewritten}. 

\begin{itemize}
	\item[(i)] Assumption~\eqref{convergence:median:1} motivates the approximation of the conditional median on Laplace scale by $m^*(x^*) \approx  a x^* + (x^*)^\beta b$, as $x^* \to \infty$. 
	\item[(ii)] Computing $x^*$ in practice requires the unknown distribution function $F_X$ in the tail region above the threshold $u_X$. Since there are only few data points above this level, we use the semi-parametric GPD approximation $\tilde F_X$ in \eqref{eq:tail:approximation:tilde:1} to obtain the approximate transformed predictor value $ \tilde x^* \approx  x^*$.
 	\item[(iii)] If $a > 0$, the regressed median will be in the upper tail of the Laplace distribution of~$Y^*$. Therefore, the back-transformation to the original margin $F_Y$ requires the GPD approximation $\tilde Q_Y$ to have an accurate extrapolation to quantiles beyond $u_Y = Q_Y(\tauzero)$. If $a < 0$, a similar approximation holds for the lower tail of $F_Y$.
\end{itemize}

Our main result provides bounds on the relative error between $\tilde m$ and $m$ under suitable conditions that control the biases in the different approximations.

\begin{theorem}\label{thm:main}
	Let $X$ and $Y$ be random variables with infinite upper endpoint satisfying the domain of attraction condition \eqref{eq:domain:attraction:1} 
	with shape parameters $\gamma_X, \gamma_Y \geq 0$. 
	Under Assumption~\ref{main_assumption_text} with $a\in[0,1]$
	the relative error of between the approximation $\tilde m(x)$ in~\eqref{extr_pinciple} and the true median function $m(x) = \M(Y \mid X = x)$  
	vanishes uniformly for large predictors, that is, 
	\begin{align}\label{main_conv}
		\lim_{\tau_0 \to 1} \sup_{
	   \substack{
			Q_X(\tau_0)  < x < \nu(\tau_0) 
				}  } \left|\frac{\tilde m(x)}{ m(x)} - 1 \right| = 0. 
	\end{align}
	for any extrapolation boundary $\nu(\tau_0) = \nu\{ Q_X(\tau_0)\}$ satisfying \eqref{eq:extra:condition} for $F_X$. 
	If $\gamma_Y = 0$, then this statement remains true even if $r(t) = O\{\log(t)\}$ in Assumption~\ref{main_assumption_text} and the mild assumption~\eqref{ass_aU} holds for $F_Y$.

	Moreover, if $X$ and $Y$ satisfy the second-order domain of attraction condition in \eqref{eq:second:order:intro} for auxiliary functions $\Sigma_X, \Sigma_Y$ and parameters $\rho_X, \rho_Y \leq 0$, respectively, then~\eqref{main_conv} holds for any  sequence $\nu(\tau_0)$ satisfying~\eqref{eq:extra:condition:rho1} for $F_X$ and such that there exists $\eta > 0$ with
\begin{align}\label{eq:extra:condition:rho1:thm}
	\lim_{\tau_0 \to 1} \Sigma_Y \{Q_Y(\tau_0) \}\left\{
	\frac{ 1-\tau_0 }{ 1 - \theta_X(\tau_0)} \right\}^{\rho_Y + \eta} = 0,
\end{align}
where $\theta_X(\tau_0) = F_X\{\nu(\tau_0)\}$ is the probability level to where we can extrapolate. 

	If $a\in[-1,0)$, the same results hold with first- and second-order parameters of the upper tail of $Y$ replaced by those of its lower tail; see Remark~\ref{rem:lower_tail}.
\end{theorem}

	The bound in~\eqref{main_conv} is on the relative error. Extrapolation is a difficult problem, and under such light assumptions as in Theorem~\ref{thm:main}, we believe that error bounds on absolute or squared errors are only possible by further restricting the extrapolation limit $\nu(\tauzero)$. In view of classical results in extreme value theory on relative erorrs for extrapolation of quantiles \citep[][Theorem 4.3.1]{dehaan:ferreira:2006}, this is not surprising.

The scope of our extrapolation principle~\eqref{extr_pinciple} is restricted by the sequence $\nu(\tauzero)$ in Theorem~\ref{thm:main}, which tells us how far we can reliably extrapolate the conditional median function.
This sequence is constrained by the first order conditions on $X$. 
Additionally assuming second-order conditions on the tails of $X$ and $Y$, we can relax the conditions on $\nu(\tauzero)$ and require \eqref{eq:extra:condition:rho1:thm}. 
This allows us to extrapolate further in the domain of $X$. 
In particular, if both $\rho_X, \rho_Y<0$, then we can always choose $\nu(\tau_0) = \infty$, which means that the extrapolation guarantee~\eqref{main_conv} holds on the whole interval $(Q_X(\tauzero), \infty)$.
We note that condition~\eqref{eq:extra:condition:rho1:thm} is different from classical assumptions in extreme value theory such as in Lemma~\ref{lem:uniform:tail:approximation}. It combines second-order conditions on the predictor and the response to make it suitable for regression extrapolation.

Theorem~\ref{thm:main} only addresses extrapolation of the median function in the upper tail of $X$. We can use the same principle for the lower tail to obtain an extrapolation on both sides of the training domain in the predictor space.
To this end, let $m_X$ denote the median of $X$ under the training distribution $\Ptr$ and define $(X^-, Y^-)$ and $(X^+, Y^+)$ as $(X,Y)$ conditioned on the events $\{X \leq m_X\}$ and $\{X > m_X\}$, respectively.
For the lower tail of $X$ we can directly apply Theorem~\ref{thm:main} by considering the vector $(-X^-, Y^-)$, with possibly different parameters for the domain of attraction conditions and in Assumption~\ref{main_assumption_text}.
Considering $Y^+$ and $Y^-$ avoids that the tails of $Y$ are dominated by the behavior at either very small or large predictor values. Indeed, with this setup, Assumption~\ref{main_assumption_text} is more easily satisfied for both lower and upper tail of $X$.
 This extends the approximation $\tilde m(x)$ in~\eqref{extr_pinciple} in an obvious way to values $x < Q_X(1-\tauzero)$.
More precisely, we define the both-sided semi-parametric approximation of the conditional mean function by $\tilde m(x) = m(x)$ for $x \in [Q_X(1-\tauzero), Q_X(\tauzero)]$, since in this region, non-parametric methods have classical guarantees \citep[e.g.,][Section 10]{dev1996} and produce reliable predictions. For any $x\in\mathbb R$ below or above this interval, $\tilde m(x)$ is defined by the approximation~\eqref{extr_pinciple} from the extrapolation principle with parameters determined by the distributions of $(X^-, Y^-)$ and $(X^+, Y^+)$, respectively. The following result is a corollary of Theorem~\ref{thm:main}.

\begin{corollary}\label{cor:main}
	Assume that $(X^+, Y^+)$ and $(-X^-, Y^-)$
	satisfy the assumptions of Theorem~\ref{thm:main}. Denote the lower and upper endpoints for extrapolation by $\nu^-(\tau_0)$ and $\nu^+(\tau_0)$, respectively. 
	Let $\tilde m = \tilde m_{\tauzero}$ be the both-sided \texttt{progression} approximation at level $\tauzero\in(0,1)$ as discussed above. Then
	\begin{align}\label{main_conv_both_sided}
		\lim_{\tau_0 \to 1} \sup_{
	   \substack{
			x \in [\nu^-(\tau_0), \nu^+(\tau_0)]}}
			 \left|\frac{\tilde m(x)}{m(x)} - 1 \right| = 0.
	\end{align}
\end{corollary}

The convergence in~\eqref{main_conv_both_sided} is a uniform statement on the point-wise performance of the extrapolation.
It is stronger than a corresponding result on the expected test error under a new distribution, as often used in the machine learning literature on out-of-distribution generalization under covariate shifts \citep[e.g.,][]{ma:pathak:wainwright:2023}.
Indeed, our result implies a bound on the expected error under 
a test distribution $\Qtesttau$ satisfying 
\begin{align}\label{Qtest}
	\mathbb P_{\Qtesttau}\left(X \in  [\nu^-(\tau_0), \nu^+(\tau_0)]\right) = 1.
\end{align}
Note that $\Qtesttau$ can have support beyond the training observations, with boundaries given by the extrapolation limits at level $\tauzero$.

\begin{theorem}\label{thm2}
	Let $\tilde m = \tilde m_{\tauzero}$ be the both-sided \texttt{progression} approximation at level $\tauzero\in(0,1)$, and let $\Qtesttau$ be a test distribution satisfying~\eqref{Qtest}. Under the assumptions of Corollary~\ref{cor:main}, the expected relative extrapolation test error vanishes, that is,
	\begin{align}
		\mathbb E_{\Qtesttau} \left( \frac{\tilde m_{\tauzero}(X)- m(X)}{m(X)} \right) \to 0, \qquad \tauzero \to 1.
	\end{align}
\end{theorem}
The proof is a direct consequence of Corollary~\ref{cor:main} and is therefore omitted.

\section{Models}\label{sub:sec:examples}

The extrapolation guarantees for the \texttt{progression} approximation in Theorem~\ref{thm:main} rely on assumptions on the marginal tail behaviors of $F_X$ and $F_Y$, and on the parametric approximation of the conditional mean function on the Laplace scale in Assumption~\ref{main_assumption_text}.
In this section we show that for a wide range of situations these conditions are satisfied. We mainly focus on the classical additive noise model, since this is the most widely used assumption in regression. Our approach covers also other modeling assumptions, such as the pre-additive noise model introduced in~\cite{shen:meinshausen:2023} particularly for extrapolation.

Since in Section~\ref{subsec:examples:tails} we have discussed that first- and second- order domain of attraction conditions hold for most relevant distributions, the main difficulty typically is to verify the median approximation \eqref{convergence:median:1}. The results rely on convolution properties of random variables, which for general distributions are highly non-trivial and require case-by-case analysis \citep[e.g.,][]{EOW2019}. 
Also other examples that are not studied here will likely satisfy \eqref{convergence:median:1}, even if analytical verification can be difficult.
The results and examples of this section motivate that the \texttt{progression} method can be applied in practice in diverse range of scenarios for regression extrapolation.

\subsection{Additive noise models}\label{sec:add_noise}

The most widely used modeling assumption in non-parametric regression is that the response $Y$ is generated by adding an independent noise term to the signal $f(X)$, where $f: \mathbb R \to \mathbb R$ is the regression function that is applied to the univariate predictor $X$. This additive noise model is therefore specified by the equation
\begin{align}\label{add_noise}
	Y = f(X) + \eps,
\end{align} 
together with training distributions $\P_X$ and $\P_\eps$ of the predictors and the noise term, respectively. Without losing generality, we assume in the sequel that $\M(\eps) = 0$, which implies that the conditional median $m(x) = f(x)$ for all $x\in\mathbb R$.

The validity of Assumption~\eqref{main_assumption_text} for the additive noise model depends on the interplay between regularity of the regression function $f(x)$ for large values of $x$,
 and the distributions of predictor $X$ and noise $\eps$. Roughly speaking, 
 the noise $\eps$ can not have a much heavier tail than the signal $f(X)$.
  This is certainly a natural requirement  in regression tasks. 
  The most extreme case is that the noise term is equal to zero; this was discussed in Section~\ref{sec:model_assumption}.
   The next lemma gives an extension of this result to more realistic situations where the noise does not substantially alter the tail of the response. This will be helpful in several examples later on.
\begin{lemma}\label{lem:median:equation}
	Let $X$ and $Y$ be random variables with distribution functions $F_X$ and $F_Y$, respectively, admitting an infinite upper endpoint. Assume that for the median $m(x) = \M( Y \mid X = x)$ it holds that $ \bar F_Y \circ m(x) \sim (1+c) \bar F_X(x) $, as $x \to \infty$.  Then,
	\beao 
	|\M( Y^* \mid X^* = x^*) - x^* - \log\{1/(1+c)\}| \to 0, \quad x^* \to \infty.
	\eeao 
\end{lemma}

For extrapolation of the regression function $f$ beyond its training domain, it is inevitable to make certain regularity assumptions. \cite{shen:meinshausen:2023}, for instance, consider only monotone functions in their pre-additive noise model framework; see Section~\ref{engression}. \cite{pfister:buhlman:2023} assume that $f$ is several times differentiable and that derivatives are not more extreme in the test than in the training domain. We would like to avoid assumptions on differentiability and, when possible, on monotonicity. Instead we often only require $f$ to be regularly varying; a property that does not have to hold only at the boundary of the domain.
\begin{definition}
	A  measurable function $f$ is regularly varying with tail index $\alpha_f \in \mathbb{R}$ if for all $x > 1$, $\lim_{t \to \infty} f(tx)/f(t) = x^{\alpha_f}.$ In this case we write $f \in RV(\alpha_f)$.
\end{definition}

A regularly varying function $f$ with index $\alpha_f > 0$ can be thought of as behaving like a polynomial at $\infty$, since we may write
\[ f(x) = L(x) x^{-\alpha_f}, \quad x\in \mathbb R,\]
where $L:\mathbb R\to\mathbb R$ is a slowly varying function, 
that is, a regularly varying function with tail index $\alpha_L=0$.
Examples for $L$ include, for instance, constant functions, $\log(1+x)$ and $\{\log(1+x)\}^p$, for $p >0,$ and $x >0$. 
According to Corollary~\ref{cor:main}, we can treat extrapolation at the lower and upper
 boundaries of the predictor space separately. In the sequel, 
 without mentioning explicitly, 
 we will therefore only consider the part $(X^+, Y^+)$ of the model; 
 results for the lower tail model $(-X^-, Y^-)$ can be found analogously. 
 Without loss of generality, we assume that $\M(X) = 0$
 so that we only have to define $f$ on the domain $[0,\infty)$.

\subsubsection{Regularly varying tails: $\gamma_X > 0$}\label{sec:regularly:v:tails}

We start with the case where $X$ satisfies the domain of attraction 
condition with shape parameter $\gamma_X > 0$. 
This class of heavy-tailed distributions corresponds to so-called regularly varying tails, that is, whose survival functions are regularly varying with positive index $\alpha_X = 1/\gamma_X > 0$ and thus decay with a polynomial rate. 
Convolutional properties of such distributions are substantially easier than 
those for lighter tails, and therefore we can allow for more general
 noise distributions and regression functions.

 \begin{proposition}\label{example:asymptotic:dependence}
	Assume the additive noise model~\eqref{add_noise} with continuous, regularly varying regression function $f \in\text{RV}(\alpha_f)$ with index $\alpha_f > 0$.
	Suppose that $X$ satisfies the domain of attraction condition with shape parameter $\gamma_X > 0$. 
	Assume the noise $\varepsilon$ does not have heavier tails than the signal
	 $f(X)$, that is, there exists $c_\varepsilon \in [0, \infty)$ such that 
	 $\P(\varepsilon > x) \sim c_\varepsilon \P(f(X)>x),$ as $x\to\infty$. 
	Then, $\bar F_Y \circ m \sim (1+c_\varepsilon)\bar F_X$, and the relative error between the approximation 
	\begin{align*} 
	\tilde m(x) =  
	\tilde Q_Y \circ F_L[\tilde x^* + \log\{1/(1+c_\varepsilon)\}], \quad Q_X(\tauzero) <  x < \nu(\tauzero),
	\end{align*}
	and the true median function $m(x)$ converges to zero in the sense of~\eqref{main_conv}, with extrapolation boundary $\nu(\tauzero)$ corresponding to the first-order domain of attraction conditions.
\end{proposition}

 The key assumptions for extrapolation in the above results are the first-order 
 regularity conditions on the distribution functions of $X$ and $f$.
  In this case, the distribution of the noise $\varepsilon$ can be arbitrary 
  as long as it is not too heavy-tailed. 
We now consider examples satisfying the assumptions of Proposition~\ref{example:asymptotic:dependence}.

\begin{example}
	Suppose that $X$ follows a Pareto distribution $\P(X > x) = x^{-\alpha_X}$ for some $\alpha_X > 0$, and that $\varepsilon$ satisfies the domain of attraction condition~\eqref{eq:domain:attraction:1} with $\gamma_\varepsilon =0$.
	Let $f\in\text{RV}(\alpha_f)$ be a regularly varying function with $\alpha_f >0$ 
	Then $Y= f(X) + \varepsilon$ satisfies the assumptions 
	in Proposition~\ref{example:asymptotic:dependence}.
\end{example}

The class of noise distributions with $\gamma_\varepsilon =0$ contains many important examples, such as normal distributions, Laplace distributions and most other light-tailed distributions.

\begin{example}
	Suppose that $X$ and $\eps$ follow a Pareto distributions  
	with tail indices $\alpha_X>0$ and $\alpha_\varepsilon>0$, respectively, and let 
	$f\in\text{RV}(\alpha_f)$ be a regularly varying function with $\alpha_f >0$. 
	If $ 0< \alpha_X/\alpha_f < \alpha_\varepsilon$, then $Y= f(X) + \varepsilon$ 
	satisfies the assumptions in Proposition~\ref{example:asymptotic:dependence}.
\end{example}

Note that in the above example, $f(X)$ has also a regularly varying tail with shape
 parameter
 $\gamma_{f(X)} = \alpha_X / \alpha_f$; see Lemma~\ref{lem:regular:variation:fx} 
 in the Appendix.
  In the class of regularly varying distributions, 
  a smaller tail index indicates a substantially heavier tail.
   The condition $\alpha_X/\alpha_f < \alpha_\varepsilon$ therefore reflects 
   the intuition that the noise has a lighter tail than the signal in this example.

 \begin{remark}
	The extrapolation boundary $\nu(\tauzero)$ in~\eqref{main_conv} is determined by the first-order conditions on $X$ and $Y$. Guarantees for extrapolation on a larger interval can be obtained  by showing that $X$ and $Y$ satisfy second-order conditions.
	For instance, if $X$ and $\varepsilon$ verify the second-order conditions, arguments similar to \cite{geluk:dehaan:resnick:starica,geluk:peng:vries:casper:2000}
	would imply that this property is also inherited by the response $Y= f(X) + \varepsilon$.
	We do not investigate this further since it would add another layer of technical complexity.
 \end{remark}

 \subsubsection{Subexponential tails for $\gamma_X = 0$}\label{sec:gumbel:subexponential}

The case where the domain of attraction condition~\eqref{eq:domain:attraction:1}  is satisfied with $\gamma_X=0$ includes a range of light-tailed distributions. In this section we consider for $X$ the class of subexponential distributions, which are on the heavier end of this spectrum.
The noise is assumed to have a comparable or lighter tail. For a review on subexponential distribution and their properties we refer to \cite{embrechts:kluppelberg:mikosch:2013} and \cite{cline:1986}.

\begin{definition}\label{def:subexponential}
A distribution $F$ is subexponential if for two independent realizations $X, X^\prime$ of $F$, 
\[
\lim_{t \to \infty} \frac{\P(X + X^\prime > t)}{2\P(X > t)} = 1.
\]	
\end{definition}
By Feller's Lemma \cite[Lemma 1.3.1]{embrechts:kluppelberg:mikosch:2013}, regularly varying distributions are one example of subexponential distributions with $\gamma_X> 0$. 
 The following lemma provides a subset of subexponential distributions contained in the domain of attraction with $\gamma_X = 0$. 
 
 \begin{lemma}\label{ex:exp:chi:subexponential}
	Consider a distribution $F(x) =1 - \exp\{- \chi(x) \}$, for $x \in \mathbb{R}$, where $\chi$ is a strictly monotone regularly varying function with index $ \alpha_\chi \in (0,1)$.  Then, $F$ is a subexponential distribution. 
	Moreover, if $\chi$ verifies the von Mises condition $-{\chi^{\prime \prime} (x) }/{({\chi^\prime}(x))^2} \to 0$ as $x \to \infty$ then $F$ is satisfies the domain of attraction condition with $\gamma = 0$. 
\end{lemma}

If the predictor $X$ follows a subexponential distribution as in the above lemma, then the next proposition gives a general result when the extrapolation principle is satisfied.

\begin{proposition}\label{example:asymptotic:dependence2}
	Let $X$ have a subexponential distribution as in Lemma~\ref{ex:exp:chi:subexponential} with index $\alpha_\chi  \in (0,1)$, and  
	let $f \in\text{RV}(\alpha_f)$ be strictly monotone with tail index $\alpha_f  > 0$.
	Assume there exists $c_\varepsilon \in [0, \infty)$ such that $\P(\varepsilon > x) \sim c_\varepsilon \P(f(X)>x),$ as $x\to\infty$, and suppose that 
	\[\alpha_\chi / \alpha_f < 1.\]
	Assume also that for $\chi$ and $\chi \circ f^{-1}$ satisfy the von Mises condition.
	 Then, $\bar F_Y \circ m \sim (1+c_\varepsilon)\bar F_X$, and the relative error between the approximation 
	 \begin{align*} 
	 \tilde m(x) =  
	 \tilde Q_Y \circ F_L[\tilde x^* + \log\{1/(1+c_\varepsilon)\}], \quad Q_X(\tauzero) <  x < \nu(\tauzero),
	 \end{align*}
	 and the true median function $m(x)$ converges to zero in the sense of~\eqref{main_conv}.
\end{proposition}

A large class of distributions $F_X$ satisfying the assumptions of Proposition~\ref{example:asymptotic:dependence2} are Weibull distributions. 
Recall that $X$ follows a Weibull distribution with shape and scale parameters $\tau,\lambda > 0$, respectively, if 
\[ F_X(x) = 1 - \exp\{- \lambda x^{\tau}\}, \quad x\geq 0,\]
and it is zero for $x<0$.
 We can now apply Proposition~\ref{example:asymptotic:dependence2} in the following examples. 

\begin{example}\label{ex:exponential}
Suppose that $X$ and $\varepsilon$ have Weibull distributions with parameters $\tau_X$ and $\tau_\varepsilon$, respectively. Let $f(x) =  x^{\alpha_f}$ for $x \in\mathbb R$, and $ 0 < \tau_X/\alpha_f < \min\{ \tau_\varepsilon, 1\}$. 
Then, $Y= f(X) + \varepsilon$ satisfies the assumptions of Proposition~\ref{example:asymptotic:dependence2}.
\end{example} 

Note that normal distributions do not fall in the Weibull class. We treat this case as a separate example.

\begin{example}\label{ex:Gaussian}
Let $X$ follow a normal distribution and assume  $\varepsilon$ follows a normal distribution or a Weibull distribution with parameters $\tau_\varepsilon$. 
 Let $f(x) =  x^{\alpha_f}$ for $x \in \mathbb R$,  and assume $ 0 < 2/\alpha_f< \min\{ \tau_\varepsilon, 1\}$  Then, $Y= f(X) + \varepsilon$ satisfies the assumptions of Proposition~\ref{example:asymptotic:dependence2}.
\end{example}

\subsubsection{Lighter tails for $\gamma_X=0$}\label{subsec:light:tails:gumbel}

In this section we now turn to the case where both $f(X)$ and $\varepsilon$ have lighter tails than subexponential. 
In this case, the contributions of the moderate values of both $f(X)$ and $\varepsilon$ become more important.  
For this reason, we require stronger assumptions on $X$, $\varepsilon$ and $f$ to verify the assumptions of Theorem~\ref{thm:main}.
Building on the previous work in \cite{bal1993}, and the overview in \cite{asmussen:hashorva:laub:taimre:2017}, we restrict our analysis to the following class. 
\begin{definition}
The distribution $F$ has Weibull-like tails if it admits a density $g$ such that
\beam\label{eq:den:assumptions:2}
g(x) \sim d \, x^{\eta + \tau - 1} \exp\{ - c x^{\tau}\}, \quad x \to \infty,  
\eeam	
for $\tau > 1$, $\eta \in \mathbb{R}$ and $c,d>0$.
\end{definition}

The next Proposition considers the case where $f(X)$ and $\varepsilon$ have equivalent tails.

\begin{proposition}\label{prop:weibull:light:tails}
	Assume the additive noise model~\eqref{add_noise} with $f(x) = x^{\alpha_f}$, for $x \in \mathbb R$.
	Let $X$ and $\varepsilon$ have Weibull-like tails as in~\eqref{eq:den:assumptions:2} with parameters indexed by the corresponding subscript. If $\tau_X/\alpha_f = \tau_\eps > 1$,
	the relative error between the approximation 
	\begin{align*} 
	\tilde m(x) =  
	\tilde Q_Y \circ F_L(a\tilde x^*), \quad Q_X(\tauzero) <  x < \nu(\tauzero),
	\end{align*}
	and the true median function $m(x)$ converges to zero in the sense of~\eqref{main_conv}, where $a = \theta^{\tau_\eps} +  c_{\varepsilon} c_X^{-1}  (1-\theta)^{\tau_\eps} \in (0,1)$ and  
	\beao 
	\theta = \{1 + (c_X/c_\eps)^{1/(\tau_\eps - 1)}\}^{-1}.
	\eeao
\end{proposition}
Interestingly, in this case the median on the Laplace scale does not converge to its linear approximation. Instead it suffices in Theorem~\ref{thm:main} that    
\begin{align}\label{eq:median:result}
	|\M(Y^* | X^* = x^*) - a x^*| = O\{ \log(x^*)\}, \quad x^* \to\infty,
\end{align} 
in addition with condition~\eqref{ass_aU}, which can be verified for this model.

The next examples are rather specific, but they illustrate what kind of behaviors are possible. Convolution in these light-tailed cases are difficult and therefore analytical forms are rare.

\begin{example}
Let $X$ and $\varepsilon$ follow Weibull distributions with parameter $\tau_X = 3$ and $\tau_\varepsilon = 3/2$, respectively. 
For $f(x) = x^2$, for $x \in\mathbb R$, $Y= f(X) + \varepsilon$ satisfies the assumptions of Proposition~\ref{prop:weibull:light:tails}. 
\end{example}

\begin{example}\label{examples:gaussian:copula}
	Let $X$ and $\varepsilon$ have centered normal distributions with variances $\sigma_X^2 > 0$ and $\sigma_\varepsilon^2$, respectively.
		Let $f(x) = \rho x$, for $\rho > 0$ and $x \in \mathbb{R}$. Then, $Y= f(X) + \varepsilon$ satisfies the assumptions of Proposition~\ref{prop:weibull:light:tails} with
	\[a = \frac{\rho^2\sigma_X^2}{\sigma_\varepsilon^2  +\rho^2\sigma_X^2} \in (0,1).\] 
\end{example}

For light-tails, if $f(X)$ and $\varepsilon$ do not have asymptotically equivalent tails, the next example shows that the first-order linear approximation term in~\eqref{convergence:median:1} is trivial, that is, $a=0$.  

\begin{example}\label{example:beta:not:zero}\label{lem:weibull:light:tails}
Assume the additive noise model~\eqref{add_noise} with $f(x) = x^{\alpha_f}$ for $x \in \mathbb R$. Let $X$ and $\varepsilon$ follow symmetric distributions on $\mathbb R$ with $F_X(x) = 1-\exp(-x^{\tau_X})/2$  and $F_\varepsilon(x) = 1- \exp(- x^{\tau_\eps})/2$, for $x > 0$ and $\tau_f, \tau_\eps > 1$. 
 Define $\tau_f := \tau_X/\alpha_f$ and assume that $\tau_X/\tau_f \not = 1$. Then, 
\begin{align}\label{eq:median:result:2}
		|\M(Y^* | X^* = x^*) - (x^*)^{\min\{\tau_\eps/\tau_f,1\} }| = O\{\log(x^*)\}, \quad x^* \to\infty,
	\end{align} 
	and the relative error between the approximation 
	\begin{align*} 
	\tilde m(x) =  
	\tilde Q_Y \circ F_L[(\tilde x^*)^{\min\{\tau_\eps/\tau_f,1\}}], \quad Q_X(\tauzero) <  x < \nu(\tauzero),
	\end{align*}
	and the true median function $m(x)$ converges to zero in the sense of~\eqref{main_conv}.

\end{example}

\subsection{Pre-additive noise models and \texttt{engression}}\label{engression}

The difficulty of extrapolation beyond the training data in additive noise models $Y = f(X) + \eps$ as in~\eqref{add_noise} is that the conditional distribution of $Y\mid X=x$ contains only information on the regression function~$f$ evaluated at point~$x$.
An alternative modeling assumption, that is in some sense more suitable for extrapolation is the pre-additive noise model introduced in \citep{shen:meinshausen:2023}, given by
\begin{align}\label{model_engression}
	Y = g(X + \eta),
\end{align} 
where $X$ is a univariate predictor and $\eta$ is an independent noise term. The function $g:\mathbb R\to\mathbb R$ is assumed to be strictly monotone, and therefore the inverse~$g^{-1}$ exists and we observe 
\begin{align}\label{engression_structure}
	\P(Y > y \mid X = x) &=  \bar F_\eta\{g^{-1}(y) - x\}, \qquad 	Q_{Y\mid X = x}(\alpha) = g\{x + Q_\eta(\alpha)\}.
\end{align} 
In other words, the conditional distribution and quantile function of $Y \mid X=x$ contain information on the function $g$ evaluated not only at point $x$, but at many other points. In particular, if the noise $\eta$ is unbounded, then knowledge of the quantile function for one $x\in \mathbb R$ identifies the function $g$ on its whole domain. \cite{shen:meinshausen:2023} exploit this fact by a distributional regression method called \texttt{engression}, which allows for extrapolation under the model assumption~\eqref{model_engression}.

Our \texttt{progression} method relies on a more generic extrapolation principle, but it is nevertheless interesting to see whether Assumption~\eqref{main_assumption_text} is satisfied under this pre-additive noise model; clearly, a method like \texttt{engression} that exploits the particular structure~\eqref{engression_structure} will in this case be more effective.
Note that the conditional median for this model is $m(x) = g\{x + \M(\eta)\}$
We therefore have for pre-additive noise models the relation 
\[\bar F_Y \circ m(x) = \bar F_{X+\eta-m_\eta}(x), \quad x \in\mathbb R, \]
where $m_\eta = \M(\eta)$. This observation enables the direct application of Lemma~\ref{lem:median:equation} to obtain the following result.
\begin{proposition}
	Let $X$ be random variable with distribution functions $F_X$ and infinite upper endpoint.
	Assume the pre-additive noise model~\eqref{model_engression} with strictly monotone, unbounded function $g:\mathbb R\to \mathbb R$. If $\bar F_{X+\eta - m_\eta} \sim (1+c_\eta) \bar F_X$, then Assumption~\ref{main_assumption_text} is satisfied with
	\beao 
	|\M( Y^* \mid X^* = x^*) - x^* - \log\{1/(1+c_\eta)\}| \to 0, \quad x^* \to \infty.
	\eeao 
\end{proposition}
It is now easy to check a whole range of examples for distributions of $X$ and $\eta$ where the \texttt{progression} method is applicable and the guarantees of Theorem~\ref{thm:main} hold. For instance, we can follow Propositions~\ref{example:asymptotic:dependence} and~\ref{example:asymptotic:dependence2} in the heavy- and light-tailed cases, respectively. We do not discuss details or concrete examples here, but refer to Section~\ref{sec:add_noise}.

\section{Statistical implementation and extensions} \label{sec:statistics}

The approximation~\eqref{extr_pinciple} underlying the \texttt{progression} method involves several unknown parameters that have to be estimated from the data, both concerning the marginal tails and the dependence between predictor and response. In this section, we first briefly recall the semi-parametric statistical implementation of the GPD approximations in Section~\ref{sec:first_order}. 
Since this is a classical problem in extreme value statistics, 
we refer to the pertinent literature for more details and other approaches \citep[e.g.,][]{coles:2001}.

Our \texttt{progression} can be combined with many regression methods for univariate predictors to improve extrapolation properties beyond the training data. 
We concentrate here on random forests since it turns to be well suited 
for a smooth transition from a non-parametric fit inside the sample range to a parametric extrapolation~\eqref{extr_pinciple} at the boundary regions of the training data.
We then discuss additive model \texttt{progression}, an extension of our extrapolation principle to multivariate predictor vectors $\mathbf X$ when 
the regression function has an additive structure.

\subsection{Univariate tail estimation}\label{sec:stat_univariate}

As in Section~\ref{sec:first_order}, let $Y$ be
a generic univariate random variable with distribution
 function $F$, which can either be one of the predictors or the response.
The extrapolation principle in Section~\ref{sec:extrapolation:principle} underlying the \texttt{progression} method involves approximations of univariate distributional tail probabilities and quantiles.   
Following the standard semi-parametric approach of extreme value statistics, we use the empirical distribution function of the bulk of the data and GPD approximations for the lower and upper tails. 
In particular, for independent data $Y_1,\dots, Y_n$ from $F$, we define the lower and upper thresholds as the empirical quantiles $l = \widehat Q(1-\tauzero)$ and $u = \widehat Q(\tauzero)$ at level $\tau_0 = 1-k/n$, for some $k<n$. Note that these thresholds can be expressed as the $l = Y_{k:n}$ and $u = Y_{(n-k):n}$, where $Y_{i:n}$ is the $i$th order statistic of the sample. 
 The integer $k$ is the effective sample size of exceedances used for estimation of the GPD, and for predictive tasks it can be seen as a tuning parameter.
The semi-parametric estimate of the distribution function at level $\tauzero$ is given by 
\begin{eqnarray}\label{eq:F}
	\widehat{F}(x) = \begin{cases}
		\dfrac{1}{n}\displaystyle\sum\limits_{i=1}^n 1(Y_{i} \leq x)  ,  \quad &\text{ if }    x \in (l,u] , \\
		1 - \dfrac{k}{n} \left(1 + \widehat{\gamma}_u \dfrac{x - u}{{\widehat{\sigma}_u}}\right)^{-1/\widehat{\gamma}_u}  \quad &\text{ if }  x > u,  \\
		\dfrac{k}{n} \left(1 + \widehat{\gamma}_l \dfrac{l - x }{{\widehat{\sigma}_l}}\right)^{-1/\widehat{\gamma}_l}  \quad &\text{ if }  x \leq l,  
	\end{cases}
\end{eqnarray}
where $(\widehat \sigma_l, \widehat \gamma_l)$ and $(\widehat \sigma_u, \widehat \gamma_u)$ are estimates of the parameters of the GPD distribution in the lower and upper tail, respectively.
Statistical inference for GPD parameters is one of the most studied problems in extreme value theory, and numerous well-established method exist; see \cite{bel2023} for a review. For our purpose, a maximum likelihood estimator with the restriction that $\widehat \gamma_l, \widehat \gamma_u \geq 0$ is most suitable. We refer \cite{coles:2001} for details on the implementation and to \cite{haan:ferreira:2006} for  consistency and asymptotic normality of this estimator. Practical consideration concerning the choice of the intermediate threshold or, equivalently, the number of order statistics $k$, can be found in~\cite{cae2016}. We use the implementation of the R package \texttt{mev}.

We can apply this semi-parametric marginal transformation to obtain samples on the Laplace scale by setting 
\begin{align} \label{eq:laplace:trans:emp}
	Y^*_i = Q_L \circ \widehat F(Y_i), \qquad i= 1,\dots, n,
\end{align}
where $\widehat F$ is the semi-parametric estimate in~\eqref{eq:F} of $F$.
We can apply this transformation to the predictor and response variables in~\eqref{extr_pinciple}, or to the variables conditioned to be below or above the median of $X$ in Corollary~\ref{cor:main}.

\subsection{Random forest \texttt{progression}}\label{rf_progression}

A first, straightforward statistical implementation of the extrapolation principle requires to choose the limits $l_X$ and $u_X$ in the predictor space 
(e.g., the same thresholds as in the semi-parametric marginal normalization for $F_X$ described above), 
and then fit a traditional smoother 
inside the bulk $(l_X,u_X]$ and an extrapolation model based on~\eqref{extr_pinciple} outside of this interval.
In particular, for extrapolation above $u_X$, we would fit a median regression on the transformed scale
\begin{align}
	\label{naive}
	(\widehat a, \widehat \beta, \widehat b) = \underset{a\in[-1,1],\beta\in[0,1),b\in\mathbb R}{\argmin}\frac1n \sum_{i=1}^n \einsfun\{X_i > u_X\} |  Y_i^* - a X_i^* - (X_i^*)^\beta b|.
\end{align}
A difficulty in statistical estimation of this model is the parameter $\beta\in[0,1)$ appearing in the second-order, sub-linear term $(x^*)^\beta$ of the approximation in Assumption~\ref{main_assumption_text}. Since we only use the largest observations above some threshold to fit this parametric approximation, there is large uncertainty related to an estimate $\widehat \beta$. Moreover, in almost all examples studied in Section~\ref{sec:add_noise}, the second-order term vanishes. Even if that was not the case, the leading term in the approximation is the linear term $ax^*$, $a\in[-1,1]$, which dominates the extrapolation behavior.
We therefore assume in the sequel that $\beta = 0$ and that we fit simple linear models on the Laplace scale. 
\begin{remark}
	For modeling of extreme value dependence in copulas, \cite{heffernan:tawn:2004} fit a distributional regression of $Y^*$ on $X^*$ given that $X>u$. In their case, the parameter $\beta\in[0,1)$ is important as it specifies the variance of the conditional distributions.
	In our setting of regression extrapolation, we are interested in the easier problem of estimating the conditional median rather than the full distribution.
\end{remark}

The approach in~\eqref{naive} has two disadvantages. First, it requires to pre-specify a hard threshold where, on the transformed scale, the prediction function transitions from a possible non-parametric smoother 
 to a parametric model. The threshold would be a tuning parameter.
Second, it gives equal weights to all observations above the threshold, regardless of how far in the tail they are.
We propose a different statistical implementation that leverages the localizing weights of a regression random forest \citep{breiman:2001,meinshausen:2006, athey:tibshirani:wager:2019}. This results in a data-driven way of choosing the threshold above which a linear model is fitted on the transformed scale. Moreover, each training observation receives its own adaptive weight to quantify the importance for extrapolation. This tree-based method provides simultaneously a good fit in the bulk and the tail of the predictor distribution, without the need of tuning the threshold from which extrapolation starts.

More precisely, we solve the localized $L_1$-optimization problem 
\begin{align}
	\label{llf}
	(\widehat a(x^*), \widehat c(x^*)) = \argmin_{a\in[-1,1],c \in\mathbb R} \sum_{i=1}^n w_i(x^*) \left| Y_i^* - c  - (X_i^* - x^*) a \right|,
\end{align}
where $w_i(x^*)$ is the localizing weight of a standard regression random forest \citep[see, e.g.,][]{meinshausen:2006} for details. More specific weights coming from tailor-made tree splits for our problem give similar results, and we therefore use the existing random forest implementation in the R package \texttt{grf}.  
The optimization problem in~\eqref{llf} should be understand as follows: a linear median regression is fitted locally at $x^*$, where the localization is achieved through the weights $w_i(x^*)$, which give more importance to training observations close to $x^*$. The coefficient $\widehat a(x^*)$ is the estimated slope of this regression line and corrects for the local linear trend in $X^*_i - x^*$. The parameter $\widehat c(x^*)$ is the value of the estimated linear function at $x^*$ and therefore an estimated of the conditional median function $\M(Y^*\mid X^* = x^*)$ on the transformed scale; see \citet[][Section 6.1]{has2009} for background on kernel smoother and local linear regression. 
A related idea was proposed in the context of local linear forests in \cite{friedberg:tibshirani:wager:2021},
but with the purpose of improving the fitting of smooth signals in classical regression random forests. There are two main differences to the local linear forest: first, the data is first transformed to the Laplace scale to make sure make use of our extrapolation principle; and second, we fit the localized median instead of the mean to guarantee estimation of the conditional median after transforming back to the original scale.

Importantly, for any $x, \bar x > X_{n:n}$ above the largest predictor observation, we have that $w_i(x^*) = w_i(\bar x^*)$ for all $i=1,\dots,n$. A close look at~\eqref{llf} then shows that above $X_{n:n}$, the extrapolation is linear on the transformed scale as suggested by the extrapolation principle; indeed, we can check that $a(\bar x^*) = a(x^*)$ and $c(\bar x^*) = c(x^*) + (\bar x^* - x^*)a(x^*)$. We call the regression method in~\eqref{llf} the random forest \texttt{progression} method, since it combines classical random forests with our extrapolation principle. It decides dynamically and in a data-driven way when the transformed data starts to be linear, and thus, where extrapolation begins.

\subsection{Additive model \texttt{progression}}\label{add_progression}

The \texttt{progression} approximation of the extrapolation principle for univariate predictors as discussed in Section~\ref{sec:extrapolation_principle} can be extended to $p$-dimensional predictor vectors $\mathbf X = (X_1,\dots, X_p)$ in several ways. We consider here the framework of additive models, where the response is
\begin{align}
	\label{add_model} 
	Y = \alpha + \sum_{j=1}^p f_j(X_j) + \varepsilon,
\end{align}
with regression functions $f_j: \mathbb R \to \mathbb R$ for $j=1,\dots, p$.
The idea is to apply extrapolation by \texttt{progression} for each function separately and then combine them to the final estimate.

Suppose that we have $n$ samples $(\g X_1, Y_1), \dots, (\g X_n, Y_n)$ of  response $Y$ and predictors $\g X = (X_1, \dots, X_p)\in\mathbb R^p$.
A first, naive approach would be to use the classical backfitting algorithm \citep[e.g.,][]{has1987}, an iterative method to fit additive models, with a suitable class of smoothing functions to obtain estimates $\widehat f_j$ for  every regression function $f_j$, $j=1,\dots, p$. For the $j$th regression function one could then compute the residuals 
\[R_{ij} = Y_i - \widehat \alpha - \sum_{k\neq j} \widehat f_j(X_{ik}), \quad i=1,\dots,n,\]
where $\widehat \alpha$ is the sample mean of the response.
If the smoothing estimates $\widehat f_k$ were good approximations of their true counterparts, then the residuals would be approximate samples of the univariate regression model
\begin{align}\label{uni_additive}
	R_j = f_j(X_j) + \varepsilon.
\end{align}
In this case, a \texttt{progression} method for a univariate predictor could be applied to the sample $(X_{ij},R_{ij})$ to obtain an improved estimate of $f_j$ with extrapolation guarantees as in Theorem~\ref{thm:main}; see, for instance, Section~\ref{rf_progression} for such a method. However, since classical smoothing functions typically do not provide accurate estimates at the boundary of the training sample, the residuals $R_{ij}$ for the smallest and largest predictor samples $X_{ij}$ will be subject to high uncertainties. Since extrapolation is particularly sensitive to those observations, in Algorithm~\ref{algo:backfittng} we propose an adaptation of the backfitting that applies \texttt{progression} in every step of the iterative procedure.   
In particular, each time the $j$th variable is selected in the iteration, \texttt{progression} is applied to the residuals as new response, that is, they are transformed to the Laplace scale to regress $R^*_{ij}$ on $X^*_{ij}$, $i=1,\dots, n$.

\begin{algorithm}
	\caption{ \textsc{\scriptsize BACKFITTING ALGORITHM FOR PROGRESSION} }
	{\bf Input:} Independent samples $(\bfX_1, Y_1), \dots, (\bfX_n, Y_n)$ from an additive model~\eqref{add_model}.\\
	{\bf Output:} Estimates of all regression functions $\widehat{f}_j$, $j=1,2,\dots, p$.\\
	\quad Intialize $\widehat f_j(x) = 0$ for all $x\in\mathbb R$, $j=1,\dots,p$ and $\widehat \alpha = \frac1n \sum_{i=1}^n Y_i$.\\
	\quad Cycle through $j=1,2,\dots,p,1,2, \dots, p,1,2,\dots$ \\
	\quad\quad\quad Compute residuals \[R_{ij} = Y_i - \widehat \alpha - \sum_{k\neq j} \widehat f_j(X_{ik}), \quad i=1,\dots,n.\]\\
	\quad\quad\quad Fit \texttt{progression} smoother $S$ to data set $(R_{1j}, X_{1,j}), \dots, (R_{nj}, X_{n,j})$ and set 
	\[ \widehat f_j(X_{ij}) = S(R_{ij}).\]
	\quad Until convergence of the functions $\widehat f_j$.
	 \label{algo:backfittng}
\end{algorithm}

We discuss briefly the intuition why the backfitting algorithm is able to extrapolate the response in the presence of multi-dimensional predictor vectors.
In order to fit a univariate \texttt{progression} to the $j$th residual model~\eqref{uni_additive} that has extrapolation guarantees from Theorem~\ref{thm:main}, the approximate residuals $R_{ij}$ computed in the inner loop of the backfitting algorithm should be sufficiently accurate, especially for predictor values $X_{ij}$ at the boundary of the training domain of the $j$th predictor. 
In the additive model~\eqref{add_model}, the tails of the response distribution $Y$ are determined by the contributions to large observations of each of the regression terms $f_j(X_j)$, and the noise variable $\eps$. Under the usual assumption that the noise is not the dominating term, the upper tail of $Y$ will be dominated by the random variable $f_{j_0}(X_{j_0})$ with the heaviest upper tail. 
The first regression functions that is well extrapolated is therefore $f_{j_0}$, since the assumption of Theorem~\ref{thm:main} are satisfied for the pair $(X_{j_0}, R_{j_0})$. Once the influence of $f_{j_0}(X_{j_0})$ on the tail of $Y$ is subtracted in the other residuals, the next heavier regression term is extrapolated correctly, and so on.

The backfitting algorithm can be used with any smoother that has extrapolation guarantees. In the sequel we use the random forest \texttt{progression} from Section~\ref{rf_progression}, because of the advantages discussed in the previous section and since it requires little tuning.

\section{Experiments}\label{sec:experiments}

We perform several simulation experiments to evaluate the extrapolation properties of our \texttt{progression} methods, where throughout this section we use random forest \texttt{progression} from Section~\ref{rf_progression}.

We start with a univariate predictor $X$ and sample from additive noise models~\eqref{add_noise} for different regression functions $f: \mathbb R \to \mathbb R$ and different training distributions $\Ptr$ for the predictor and noise. For the form of the regression function $f$, we follow some of the examples considered in \texttt{engression} \citep{shen:meinshausen:2023}. Namely, we consider a two-sided square-root function, a quadratic function and a cubic function, in order to illustrate different extrapolation behaviors.
\citep{shen:meinshausen:2023} simulate the predictors as a uniform distribution on a compact interval, resulting in hard lower and upper endpoints. For extrapolation problems in applications such as environmental or climate science, it is more relatistic to allow for predictors with unbounded lower and upper endpoints. We therefore sample $X$ from a $t$-distribution with $3$ degrees of freedom. To be consistent with our theory, the distribution of the noise $\varepsilon$ can not have heavier tails than the signal $f(X)$. We therefore choose a normal distributions in our simulations, but many other examples are possible.

We generate training data sets with $n=1000$ samples from $(X,Y)$ under the distribution $\Ptr$ and fit a random forest, a local linear forest, \texttt{engression} and our random forest \texttt{progression} described in Section~\ref{rf_progression}. 
For \texttt{engression} we use the default implementation and run it for 1000 epochs. The intermediate threshold probability for the semi-parametric estimation in~\eqref{eq:F} of \texttt{progression} is chosen as $\tauzero = 0.9$ or, equivalently, $k=100$; different values for $\tauzero$ in the range $[0.8,0.95]$ give similar results.
As test points, we deterministically fix a sequence of equidistant points with values in $[Q_X(0.01/n),Q_X(1-0.01/n)]$, where $Q_X$ is the quantile function of $X$.
Note that these boundaries are indeed out-of-sample, since only one in 10000 samples of the training distribution $\Ptr$ would lie above or below this interval.

Figure~\ref{fig:sim1} show the results for 50 repetitions of the fitting. We first observe that a classical random forest does not extrapalote well since it predicts constants outside of the training data. The local linear forest extrapolates linearly, which improves the fit on all functions significantly; better prediction in sparse regions was already highlighted in \cite{friedberg:tibshirani:wager:2021}. However, in the quadratic and cubic functions we can see the limitations of linear extrapolation. Both \texttt{engression} and \texttt{progression} adaptively estimate the shape of the required extrapolation. On these three functions, they perform similarly. While \texttt{engression} seems to underestimate the slope in the cubic function, \texttt{progression} overestimates in the quadratic function.

\begin{figure}[tb]
	\begin{center}
	\includegraphics[width=1\linewidth]{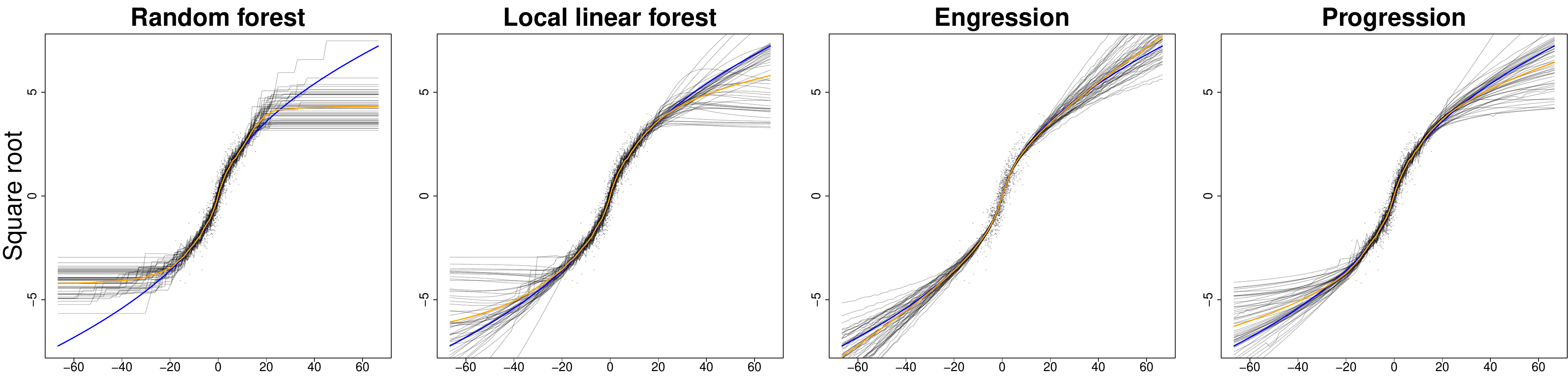}
	\includegraphics[width=1\linewidth]{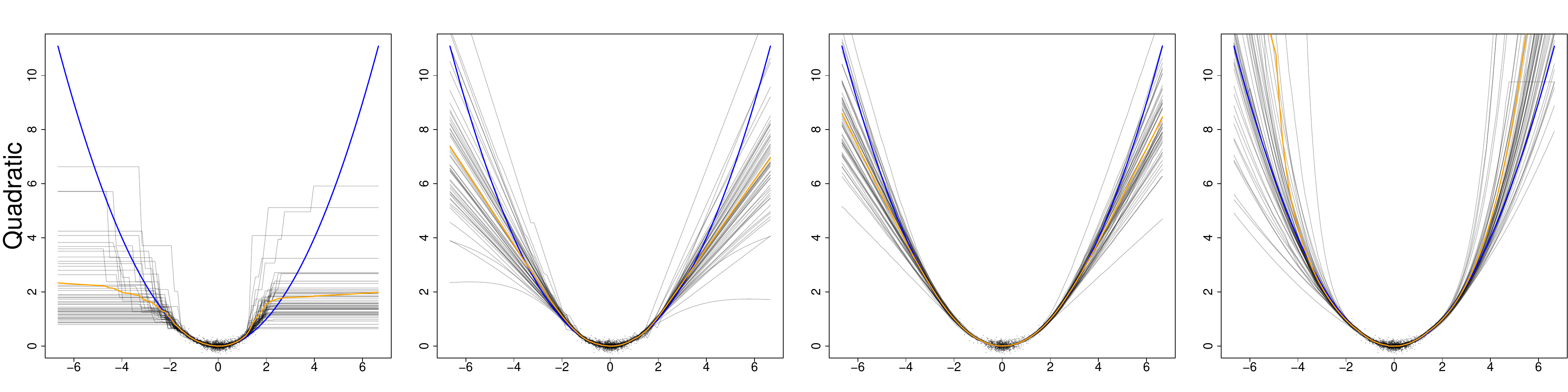}
	\includegraphics[width=1\linewidth]{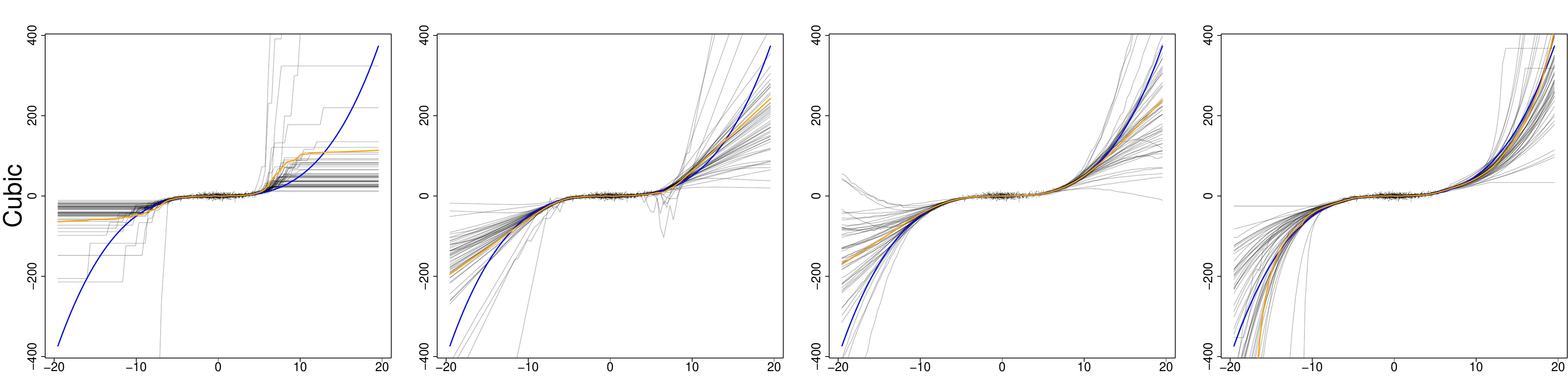}
	\end{center}
	\caption{50 fits (gray lines) of the four methods on the three univariate regression functions, together with the true regression function $f$ (blue line), the mean of the 50 gray fits (orange line), and the training data (gray points) from one repetition.
	 }	\label{fig:sim1}
	\end{figure}

\begin{figure}[h!]
	\begin{center}
	\includegraphics[width=.45\linewidth]{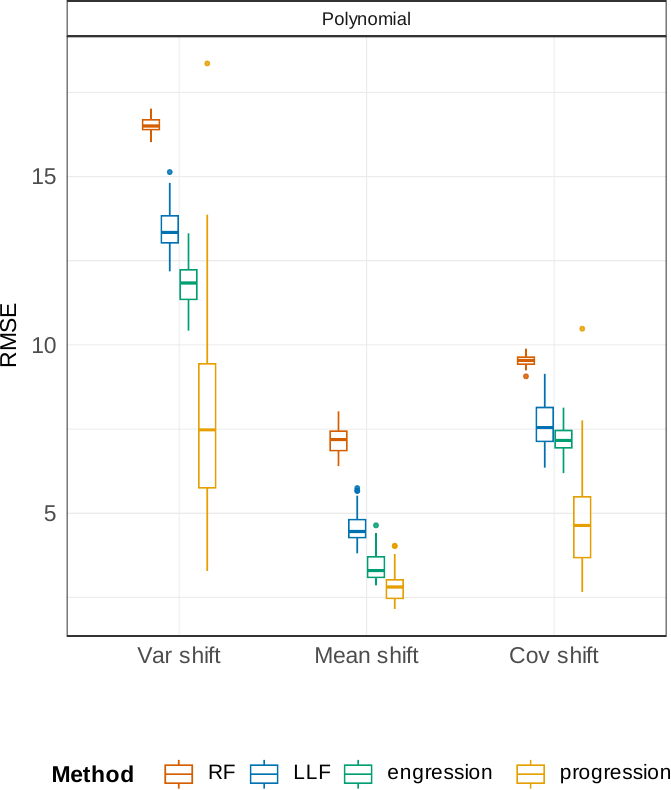}
	\includegraphics[width=.45\linewidth]{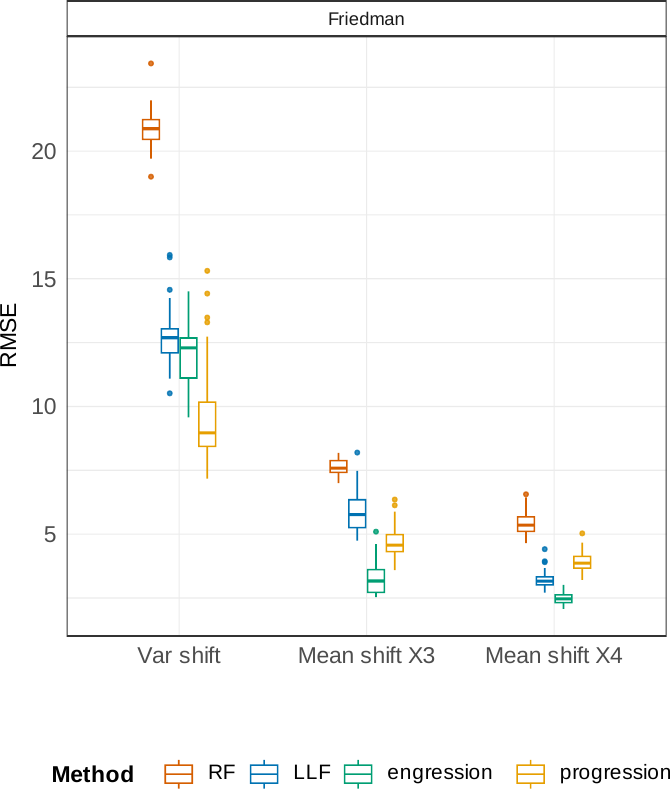}
	\end{center}
	\caption{Boxplots of root mean square errors (RMSE) of the four different methods under different out-of-distribution test samples from $\Pte$ for the fractional polynomial~\eqref{poly} (left) and the Friedman function~\eqref{friedman} (right).}
	\label{fig_rmse}
\end{figure}

We now consider the multivariate setting, where $\mathbf X$ is $p$-dimensional predictor vector. We compare our additive model \texttt{progression} as described in Section~\ref{add_progression} with a classical regression forest, a local linear forest and \texttt{engression}. 
The first model we simulate from is an additive model in $p=2$ dimensions with fractional polynomials as marginal functions, that is,
\begin{align}\label{poly}
	Y = \text{sign}(X_1) 4 \sqrt{|X_1|} + X_2^3 / 20 + \varepsilon,
\end{align}
where under the training distribution $\Ptr$, we generate $n=1000$ samples of $(X,Y)$ with $\varepsilon \sim N(0, 2^2)$, and where $\mathbf X$ is distributed according to a bivariate $t$-distribution with 4 degrees of freedom and independent components. 
In order to evaluate the extrapolation performance of the four methods, we consider different covariate shifts for the test distribution $\Pte$ that go significantly beyond the data range, namely
\begin{itemize}
	\item[(i)] mean shift: we shift the mean of the second component of $X$ by a constant $c$, that is, $\Pte( \mathbf X\in \cdot ) = \Ptr(\mathbf X + (0,c)\in \cdot)$;
	\item[(ii)] variance shift: we increase the variance by some factor $s^2$, that is, $\Pte( \mathbf X\in \cdot) = \Ptr(s\mathbf X \in \cdot)$;
	\item[(iii)] covariance shift: we change the covariance structure of the predictors, that is, under $\Pte$, $\mathbf X$ is a the same $t$-distribution but with correlation $\rho = 80\%$ in the parameter matrix. 
\end{itemize}

The left-hand side of Figure~\ref{fig_rmse} shows boxplots of the root mean square error (RMSE) of the four different methods evaluated the 200 samples of $(X,Y)$ under the test distribution $\Pte$ with the three domain shift described above.  
We see that under all scenarios, the classical random forest performs worst. This highlights that constant extrapolation in not suitable for general regression functions. Local linear forests extrapalote linearly, a fact that improves significantly its out-of-distribution performance; The two extrapolation methods \texttt{engression} and \texttt{progression} perform even better on this example. In particular, \texttt{progression} outperforms the other methods in all covariate shift settings. 

The second function with $p=5$ predictors was introduced in \cite{fri1991} and has the form
\begin{align}\label{friedman}
	Y = 10\sin(\pi X_1 X_2) + 20 (X_3 - 0.5)^2 + 10X_4 + 5X_5 + \varepsilon,
\end{align}
where $\varepsilon \sim N(0, 2^2)$. This model is popular for evaluating non-parameteric regression methods, since it tests how well interactions are handled, how it picks up the quadratic signal and how it captures the strong linear signals; see for instance \cite{chi2010}. In the local linear forest
paper \cite{friedberg:tibshirani:wager:2021}, the predictor vector $\mathbf X$ was sampled from a uniform distribution on the 5-dimensional hypercube. To make the problem interesting for extrapolation, we sample $\mathbf X$ instead from a zero-mean, multivariate $t$-distribution with 4 degrees of freedom and independent components. We consider similar covariate shifts for the test distributions $\Pte$ as above, including a variance shift, and mean shifts in $X_3$ and $X_4$.
Note that the model~\eqref{friedman} is not additive and our additive \texttt{progression} is thus misspecified. The right-hand panel of Figure~\ref{fig_rmse} shows the results for extrapolation performance on this function. Again, \texttt{engression} and \texttt{progression} outperform random forest and local linear forest. Interestingly, when the mean shift is in $X_4$, a very strong linear predictor, then local linear forest has a fairly low RMSE. This is expected, since the linear extrapolation of this model is in this case the perfect model assumption. In this setting, \texttt{engression} performs best on the two mean shift test ditributions.

\bibliographystyle{Chicago}
\bibliography{ref}

\appendix

\section{Preliminaries on extended regular variation}\label{sec:appendix}

\subsection{Extended regular variation}\label{sec:app:1}

The following auxiliary functions will be useful for proving our results.  
\beam\label{eq:notation}
U(t) = Q(1 - 1/t), \qquad a(t) = \sigma\{U(t)\}, \qquad  A(t)  = \Sigma\{U(t)\},
\eeam 
for $t\in\mathbb{R}$. 
The proof of Theorem~\ref{thm:main} relies on the theory of regularly varying and extended regularly varying functions, which studies the approximation of the monotone function $U(x)$ as $x \to \infty$. We review below relevant definitions and important results that we use to prove Theorem~\ref{thm:main}.

By \cite[Theorem 1.1.6]{dehaan:ferreira:2006}, the relation in~\eqref{eq:domain:attraction:1}
is equivalent to the first-order condition stated in terms of the tail quantile function 
\begin{equation}\label{eq:first:order:U}
	\lim_{t \to \infty} \frac{U(tx) - U(t)}{a(t)} = \frac{x^\gamma - 1}{\gamma}, \quad x >0,
\end{equation}
and when $\gamma = 0$ the right-hand side should be read as the limit as $ \gamma \downarrow 0$, 
yielding $\log x$. 
In this sense, the above extrapolation is a first-order approximation of the tail of $Y$. 
We can thus the derive the following properties relating the functions $U$ and $a$. 
By \citet[][Lemma 1.2.9]{dehaan:ferreira:2006}, we have that if $\gamma > 0$,
then $U(t) \to \infty$ and $U(t)/a(t) \to 1/\gamma$, as $t \to \infty$, and if $\gamma = 0$, then $a(t)/U(t) \to 0$, as $t \to \infty$.

A consequence of \eqref{eq:first:order:U} that follows from \cite[Proposition B.2.17]{dehaan:ferreira:2006} together with \cite[Theorem B.2.18]{dehaan:ferreira:2006} is given in the next proposition.
\begin{proposition}\label{lema:potters:bounds}
	Assume $Y$ satisfies the first-order domain of attraction condition \eqref{eq:first:order:U} for $\gamma \geq 0$.   Then, there exists a positive function $a_0$ such that $a_0(t)/a(t) \to 1$, as $t \to \infty$, and for all $ \eta, \delta > 0$ there exists $t_0 = t_0(\eta,\delta) > 0$, such that for $t , tx \geq  t_0 $,
	\beam\label{eq:potter}
	\left| \frac{U(tx) - U(t) }{a_0(t)} - \frac{x^\gamma - 1}{\gamma} \right| \leq \delta \max\{ x^{\gamma +\eta}, x^{\gamma - \eta}\}.   
	\eeam  
\end{proposition}

To derive sharper bounds on the tail quantile function it is customary to assume second-order conditions as \eqref{eq:second:order:intro}.
In fact, \citet[Theorem B.3.19]{dehaan:ferreira:2006} implies that~\eqref{eq:second:order:intro} is equivalent to
\begin{eqnarray}\label{eq:second:order:condition}
	\lim_{t \to \infty} \{A(t)\}^{-1} \left(\frac{U(tx) - U(t)}{ a(t) } -  \frac{x^\gamma - 1} {\gamma}\right)
	= k_{\gamma,\rho}(x)  , \quad x>0,
\end{eqnarray} 
where $k_{\gamma,\rho}$ is given by
\begin{align}\label{eq:krho}
	k(x) = k_{\gamma,\rho}(x) &=
	\begin{cases}
		\frac{1}{\rho}
		\left( 
		\frac{ x^{\gamma+\rho} - 1}{\gamma + \rho} 
		- \frac{x^\gamma-1}{\gamma} \right), & \rho \not = 0 ,\\
		\frac{1}{\gamma}
		\left( 
		x^{\gamma}\log x 
		- \frac{x^\gamma-1}{\gamma}\right), & \rho = 0, \gamma \not = 0 ,\\
		\frac{1}{2} (\log x)^2, & \rho = \gamma = 0,
	\end{cases}
\end{align}
and $\lim_{t \to \infty} A(tx)/A(t) = x^{\rho}$.
The limiting function in~\eqref{eq:second:order:intro} 
can then be written as 
\begin{align}\label{eq:hrho}
	h_{\gamma, \rho}(x) = - \{\bar H_\gamma(x)\}^{\gamma - 1} k_{\gamma, \rho}\{\bar H_\gamma(x)\}, \quad {x >0 }.
\end{align}
Moreover, assuming only that the previous limit exists, we can derive some immediate consequences.
Notably, that the function $A$ must be regularly varying with index $\rho \leq 0$, whose sign is constant eventually.

To verify first- and second-order conditions, it is often useful to rely on the so-called von Mises' condition. 
Assume $U$ admits first and second order derivatives $U^\prime$ and $U^{\prime \prime}$, respectively, and assume the second-order von Mises condition
\beam \label{eq:von:mises}
\lim_{t \to \infty} \frac{t U^{\prime \prime}(t)}{U^\prime(t)} = \gamma - 1.
\eeam 
Applying Corollary 1.1.10 in \citep{dehaan:ferreira:2006}, we can conclude that  the function $U$ satisfies \eqref{eq:first:order:U} with $a(t) = t U^\prime(t)$, $\gamma$ as in \eqref{eq:von:mises}, and by Theorem 2.3.12 in \cite{dehaan:ferreira:2006} it is possible to choose the scale function as 
$A(t) = 1-\gamma + tU^{\prime \prime }(t)/U^\prime(t)$. 

\begin{remark}
We use of the second-order von Mises assumption in \eqref{eq:von:mises} to compute the values of  $\sigma(u)$ and $\Sigma(u)$ in the examples of Section~\ref{subsec:examples:tails} upon knowledge of the function $U$ in~\eqref{eq:notation}.
\end{remark}

We provide below \cite[Theorem 2.3.6]{dehaan:ferreira:2006} establishing the main uniform inequalities for second-order regularly varying function.

\begin{theorem}\label{eq:thm:second:order:bounds}
	Let $U$ be a measurable function. Suppose there exist a positive function $a$, some  either positive or negative function $A$ with $\lim_{t \to \infty} A(t) = 0$, and parameters $\gamma \geq  0, \rho \leq 0$ such that 
	\begin{align*}
		\lim_{t \to \infty} 
		\{A(t)\}^{-1} 
		\left( 
		\frac{U(tx) - U(t)}{a(t)} - \frac{x^\gamma - 1}{\gamma}
		\right) = k_{\gamma,\rho}(x),
	\end{align*}
	for all $x > 0$ and $k_{\gamma,\rho}(x)$ as in~\eqref{eq:krho}.
	Then, for all $\epsilon, \eta >0$, there exists $t_0 = t_0(\epsilon,\eta)$ such that for all $t, tx \geq t_0$,
	\begin{align}\label{bound_2nd_unif}
		\left|\{A_0(t)\}^{-1} 
		\left( 
		\frac{U(tx) - U(t)}{a_0(t)} - \frac{x^\gamma - 1}{\gamma}
		\right) - \Psi_{\gamma,\rho}(x)\right| 
		&\leq \epsilon \max(x^{\gamma + \rho + {\eta}}, x^{\gamma+\rho-{\eta}}),
	\end{align}
	where $A_0$ and $a_0$ satisfy $a_0(t)/a(t) \to 1$, $A_0(t)/A(t) \to c \not = 0$, as $t \to \infty$, and 
	\begin{align*} \Psi_{\gamma,\rho}(x) &=
		\begin{cases} 
			\frac{ x^{\gamma+\rho} - 1}{\gamma + \rho}, & \rho < 0 ,\\
			\frac{1}{\gamma}
			x^{\gamma}\log x, & \rho = 0, \gamma > 	0 , \\
			\frac{1}{2} (\log x)^2, & \rho = \gamma = 0.
		\end{cases}
	\end{align*}
\end{theorem}

From \eqref{bound_2nd_unif} in Theorem 3
we obtain the following quantile tail approximations.   
First, if $\gamma > 0$, $\rho  = 0$, then, for all $\eta>0$, $t > t_0(\eta)$, $x\geq 1$,
\begin{align}\label{eq:ftx:expansion}
	U(tx) = U(t) + a(t)\frac{x^\gamma -1}{\gamma} \left\{
	1+ O(A(t)x^\eta )\right\}, \quad t \to \infty.
\end{align}
Instead if $\gamma = 0$, and $\rho = 0$, then
\begin{align}\label{eq:ftx:expansion:rho0}
	U(tx) = U(t) + a(t)\log(x)\left\{1+ O(A(t)x^\eta )\right\}, \quad t \to \infty.
\end{align}
Moreover, if $\rho < 0$ then we can replace the $O(A(t)x^\eta)$ term in \eqref{eq:ftx:expansion} and \eqref{eq:ftx:expansion:rho0} simply by $O(A(t))$, such that $A(t) \to 0$ and $t \to \infty$.

\section{Proof of Lemma~\ref{lem:uniform:tail:approximation} }\label{sec:proof:uniform:tail}

\begin{proof}[Lemma~\ref{lem:uniform:tail:approximation}]
Assume that $F$ satisfies the first-order condition \eqref{eq:first:order:U} for $\gamma \geq  0$.
 Consider a sequence $\nu(u)$ such that \eqref{eq:extra:condition} holds.
 Let $ x\in (u,\nu(u))$, and define $\tau = F(x)$ and $\tauzero = F(u)$. We first show that \eqref{eq:f:tilde:approx} holds. Let
	\beam \label{eq:cond:y}
	y = \frac{1-\tauzero}{ 1-\tau}  = \frac{1-F(u)}{ 1-F(x) }   > 1,
	\eeam 
	and $\tzero = 1 / (1-\tauzero)$. We can readily check that $u = U(\tzero)$ and  $x = U(t_0 y)$.
	Consider the first-order auxiliary function $a(t)$ in \eqref{eq:first:order:U}. 
	With the notation from \eqref{eq:notation}, we have for $\gamma > 0$ that
	\begin{align*}
		\frac{1 - \tilde F(x) }{1-F(x)}  &=
		\frac{(1-\tauzero)}{(1-\tau)} \left(  1+  \gamma \frac{x - u }{ \sigma(u) } \right)_+^{-1/\gamma}\\
		&= y \left(  1+  \gamma \frac{ U(\tzero y) - U(\tzero)} { a(\tzero) } \right)_+^{-1/\gamma}\\
		&= \left(  1 + \frac{\gamma}{y^\gamma} \left[  \frac{ U(\tzero y) - U(\tzero)} { a(\tzero) } - \frac{y^\gamma - 1}{\gamma} \right] \right)_+^{-1/\gamma}.
 	\end{align*}
	For $\gamma = 0$, this can be understood as the limit as $\gamma \to 0$, i.e.,
	\begin{align*}
		\frac{1 - \tilde F(x) }{1-F(x)}  &= \exp \left\{- \left[  \frac{ U(\tzero y) - U(\tzero)} { a(\tzero) } - \log y \right]\right\}.
 	\end{align*}
 	Recall that by \eqref{eq:extra:condition}  together with \eqref{eq:cond:y} we have for some large enough constant $C>0$ that  
	\begin{align}\label{y_bound} 
	1 \leq \lim_{u \to \infty} \sup_{u < x < \nu(u)} y  =  \lim_{u \to \infty} \frac{1 - F(u)}{1-F(\nu(u)) } \leq C <  \infty.
	\end{align}
	Applying Proposition~\ref{lema:potters:bounds}, for any $\eta>0$, and for any $\delta >0$, there exists $t = t(\eta,\delta) > 0$ such that for all $y \geq 1$, $t_0 \geq t$, 
 	\begin{align*}
 		(1 + \gamma \delta C^{\eta} )^{-1/\gamma} 	\leq \left(  1 + \frac{\gamma}{y^\gamma} \left[  \frac{ U(\tzero y) - U(\tzero)} { a(\tzero) } - \frac{y^\gamma - 1}{\gamma} \right] \right)_+^{-1/\gamma}
 		&\leq 
 		(1 - \gamma \delta C^{\eta} )^{-1/\gamma}.
 	\end{align*} 
 		Furthermore, this implies
 	\begin{align*}
 		\lim_{u \to \infty} \sup_{ u< x < \nu(u)} \left| \frac{1 - \tilde F(x) }{1-F(x)}  - 1 \right| \leq \max\{ 1- (1 + \gamma \delta C^{\eta} )^{-1/\gamma}, (1 - \gamma \delta C^{\eta} )^{-1/\gamma} - 1\} .
 	\end{align*}		
 		Again, the previous relation should be understood as the limit as $\gamma \to 0$. Choosing $\delta >0$ small enough yields \eqref{eq:f:tilde:approx}.

		We now assume that $Y$ satisfies the second-order condition \eqref{eq:second:order:intro} with index $\gamma \geq 0$, auxiliary function $\Sigma$ and second-order parameter $\rho \leq 0$. Recall this is equivalent to condition  \eqref{eq:second:order:condition} where the function $A$ is defined in \eqref{eq:notation}. 
Consider a sequence $\nu(u)$ such that \eqref{eq:extra:condition:rho1} holds for some $\eta>0$ and let $x\in(u,\nu(u))$. Using~\eqref{bound_2nd_unif} in Theorem~\ref{eq:thm:second:order:bounds}, we now get the sharper bound for any $\tzero > t(\eta)$ large enough and for all $y \geq 1$ and some constant $c>0$ 
 	\begin{align*}
 		(1 + c |A(\tzero)| y^{\rho+\eta})^{-1/\gamma} 	\leq \left(  1 + \frac{\gamma}{y^\gamma} \left[  \frac{ U(\tzero y) - U(\tzero)} { a(\tzero) } - \frac{y^\gamma - 1}{\gamma} \right] \right)_+^{-1/\gamma}
 		&\leq 
 		(1 -  c |A(\tzero)| y^{\rho+\eta} )^{-1/\gamma}.
 	\end{align*} 
	From assumption \eqref{eq:extra:condition:rho1} we observe that
	\begin{align}
		\notag \sup_{ u< x < \nu(u)}  |A(t_0)|y^{\rho + \eta} &= \sup_{ u< x < \nu(u)} |\Sigma(u)|\left( \frac{1-F(u)}{1-F(x )} \right)^{\rho + \eta} \\
		\label{Sigma_bound} &=  |\Sigma(u)|\left( \frac{1-F(u)}{1-F(\nu(u))} \right)^{\rho + \eta} \to 0, \quad u\to\infty,
	\end{align}
	where we used that $u = U(t_0)$ and thus $A(t_0) = \Sigma\{U(t_0)\} = \Sigma(u)$; see \eqref{eq:notation} and the definition of $y$  in \eqref{eq:cond:y}. Consequently, combining everything yields
	\begin{align*}
		\sup_{ u< x < \nu(u)} \left| \frac{1 - \tilde F(x) }{1-F(x)}  - 1 \right| &\leq \sup_{ u< x < \nu(u)} \max\{ 1- (1 +  c |A(\tzero)| y^{\rho+\eta} )^{-1/\gamma}, (1 -  c |A(\tzero)| y^{\rho+\eta} )^{-1/\gamma} - 1\}\\
		& \to 0, \quad u\to\infty.
	\end{align*}	
	If $\gamma = 0$ all steps remain true where the expressions are interpreted as the limits as $\gamma \to 0$.

 Concerning the quantile approximation, consider a sequence $\theta(\tau_0)$ such that \eqref{eq:extra:condition} holds. Let $\tau \in ( \tau_0,\theta(\tau_0))$ and $y$ as in \eqref{eq:cond:y}. 
	By the definition of $\tilde Q(\tau)$ in~\eqref{eq:quantile:tail:approximation} we have 
 \begin{align*}
 \frac{ Q(\tau) - \tilde Q(\tau)} {Q(\tau) } & = \frac{ Q(\tau) - \tilde Q(\tau)} {  \sigma\{Q(\tau_0)\}}  \frac{\sigma\{Q(\tau_0)\}}{ Q(\tau)}\\
	& = \left[ \frac{ Q(\tau) - Q(\tauzero) - \sigma\{Q(\tau_0)\}\frac{y^\gamma - 1}{\gamma}} {\sigma\{Q(\tau_0)\}} \right] \frac{\sigma\{Q(\tau_0)\}}{ Q(\tau)} \\
	& = \left[\frac{ U(\tzero y) - U(\tzero)} { a(\tzero) } - \frac{y^\gamma - 1}{\gamma}\right] \frac{a(\tzero)}{U(t_0 y)} , 
\end{align*}
where we used that $U(t_0) = U(1/(1-\tau_0)) = Q(\tau_0)$, and $a(t_0) = \sigma\{U(t_0)\} = \sigma\{Q(\tau_0)\}$, which follows from \eqref{eq:notation}. Moreover, the case $\gamma = 0$, is understood as the limit as $\gamma \to 0$.
The fact that the function $U$ is regularly varying with index $\gamma \geq 0$ implies that for all $\kappa > 0$, there exists $t = t(\kappa )$, such that for all $t_0 > t$ and $y \geq 1$, such that
\beam \label{eq:sigma:q}
  \left| \frac{a(t_0)}{U(t_0 y ) }  \right| =  \left| \frac{a(t_0)}{U(t_0 ) }  \right| \left| \frac{U(t_0)}{U(t_0 y) }  \right|  \leq (\gamma +\kappa ) (1+\kappa )  y^{-\gamma + \kappa },
\eeam
where we used that $a(t) / U(t) \to \gamma$ as $t\to\infty$ \citep[Lemma 1.2.9]{dehaan:ferreira:2006} and applied the Potter bounds in \cite[Proposition B.1.9]{dehaan:ferreira:2006}. 
Consequently, by Proposition~\ref{lema:potters:bounds}, we find that for any $\eta > 0$, $\kappa > 0$ a $\tzero$ large enough that 
\begin{align}\label{eq:bound:q:sigma}
	\sup_{\tau_0 < \tau < \theta(\tau_0) }\left| \frac{ Q(\tau) - \tilde Q(\tau)} { Q(\tau)}   \right| = \sup_{\tau_0 < \tau < \theta(\tau_0) } \left|\frac{ U(\tzero y) - U(\tzero)} { a(\tzero) } - \frac{y^\gamma - 1}{\gamma}\right| \left| \frac{a(\tzero)}{ U(t_0y)} \right| \leq (\gamma + \kappa  ) C^{\gamma + \eta},
\end{align}
where we used the bounds for $y$ in~\eqref{y_bound}.
Letting $\tauzero$ go to one and then $\kappa$ go to zero yields~\eqref{eq:U:tilde:approx}.
We now assume again that $Y$ satisfies the second-order condition \eqref{eq:second:order:intro} with index $\gamma \geq 0$, auxiliary function $\Sigma$ and second-order parameter $\rho \leq 0$. 
Consider a sequence $\theta(\tauzero)$ such that \eqref{eq:extra:condition:rho1} holds for some $\eta>0$ and let $\tau \in (\tau_0, \theta(\tau_0) )$. Using Theorem~\ref{eq:thm:second:order:bounds} and the derivations in~\eqref{eq:sigma:q} and~\eqref{eq:bound:q:sigma}, 
we obtain that for any $\kappa  > 0$ such that $\kappa < \eta$ there exists some constants $c>0$, such that
\begin{align*}
	\sup_{\tau_0 < \tau < \theta(\tau_0) }\left| \frac{ Q(\tau) - \tilde Q(\tau)} { Q(\tau)}   \right| &\leq  c |A(\tzero)| \sup_{\tau_0 < \tau < \theta(\tau_0) } y^{\gamma+\rho+\kappa }
	\left|  \frac{a(t_0)}{U(t_0 y ) } \right|. 
\end{align*}
Finally we conclude there exists $\tilde c > 0$ such that
\begin{align*}
	\sup_{\tau_0 < \tau < \theta(\tau_0) }\left| \frac{ Q(\tau) - \tilde Q(\tau)} { Q(\tau)}   \right| &\leq  \tilde c |A(\tzero)| \sup_{\tau_0 < \tau < \theta(\tau_0) } y^{\rho+\eta}\\
	& = \tilde c |\Sigma\{Q(\tau_0)\}| \left\{ \frac{1-\tau_0}{1-\theta(\tau_0)} \right\}^{\rho + \eta} \to 0, \quad \tauzero \to 1,
\end{align*}
where for the last line we have used the the observation in~\eqref{Sigma_bound}, such that $Q(\tau_0) = U(1/(1-\tau_0))= U(t_0)$,
and assumption~\eqref{eq:extra:condition:rho1}.
\end{proof}

\section{Proof Lemma~\ref{lem:median:equation}}
\begin{proof}
	Recall $Y^* = Q_L \circ F_Y(Y)$, $X = Q_X \circ F_L(X^*)$, and $x^* = Q_L \circ F_X(x)$. Then, since the median function is invariant to monotone transformations we have, for all $x^* \in \mathbb{R}$,
	\begin{align*}
		\M[Y^* | X^* = x^*]  &= Q_L \circ F_Y( \M[Y | X = Q_X\circ F_L(x^*)] ) \\
		&= Q_L \circ F_Y \circ m(Q_X\circ F_L(x^*)).
	\end{align*} 
	Moreover, note that $\bar F_Y\circ m(x) \to \infty$, as $x\to \infty$. Indeed, this follows because we assumed that $F_X$ and $F_Y$ have infinite upper endpoints and also $ \bar F_Y\circ m(x) \sim (1+c) \bar F_X(x)$, as $x \to \infty$. 
	In addition recall,
	$Q_L(y) =  -\log(2(1-y))$, for $y > 1/2$. 
	The two previous relations imply there exists $x^*_0 \in \mathbb{R}$ such that for all $x^* > x^*_0$  we have 
	\beao 
	\M[Y^* | X^* = x^*] = -\log\{2\bar F_Y \circ m(Q_X \circ F_L(x^*))\}, \quad x^* > x^*_0.
	\eeao 
	Finally, using the fact that $\bar F_Y\circ m(x) \to \infty$, as $x\to \infty$, we have
	\begin{align*}
		&\M[Y^* | X^* = x^*]\\
		&= -\log\{2\bar F_Y \circ m(Q_X \circ F_L(x^*))\}\\
		&= -\log\{2\bar F_X (Q_X \circ F_L(x^*))\} - \log(1/(1+c))
		- \log\left\{   \frac{ \bar F_Y \circ m(Q_X \circ F_L(x^*))} {(1+c)\bar F_X  (Q_X \circ F_L(x^*))}  \right\} \\
		&= x^* - \log(1/(1+c)) + o(1), \quad x \to \infty.
	\end{align*}
	Here the last relation follows again since $Q_X \circ F_L(x^*) \to \infty$, as $x \to \infty$. This last equality yields the desired result. 
\end{proof}

\begin{remark}
	It is natural to wonder what would happen in the case of the conditional quantile at level $\tau$ of $Y$ given $X$: $q_\tau(x) = Q_\tau(Y|X=x)$, for $\tau \in (0,1)$, not necessarily equal to $1/2$. 
	In this case, we can see that the proof of Lemma~\ref{lem:median:equation} can be replicated to show that if $\bar F_X\circ q_\tau(x) \sim (1+c_\tau)\bar F_X(x)$, then 
	\beao 
	|Q_\tau(Y^* \mid X^* = x^*) - x^* - \log(1/(1+c_\tau))| \to 0, \quad x^* \to \infty.
	\eeao 
\end{remark}
\section{Proof Theorem~\ref{thm:main}}\label{sec:proof:thm:main}

We assume that $a \in [0,1]$ in Assumption~\ref{main_assumption_text} throughout the proof; the case $a\in[-1,0)$ is similar.
Before starting, we recall the notation from~\eqref{eq:notation}
\begin{equation*}
U_X(t) = Q_X(1 - 1/t), \qquad a_X(t) = \sigma\{U(t)\}, \qquad  A_X(t)  = \Sigma_X\{U_X(t)\},
\end{equation*}
and similarly for $Y$.
Let $\tauzero \in (0,1)$ and denote $t_0 =  1/(1-\tauzero)$, 
and $ u_X = U_X(t_0)$ and $u_Y = U_Y(t_0)$.
Let also $a_X = a_X(t_0)$ and $a_Y = a_Y(t_0)$, where here we drop the dependence on $t_0$ to simplify notation. 
While the results and assumption in Theorem~\ref{thm:main} are all stated in terms of $\tauzero$, for the proof it will be easier to state everything in terms of $t_0$ and $u_X$. In particular, we emphasize that 
$u_X = Q_X(\tauzero)$ and $\nu(u_X)= \nu\{Q_X(\tauzero)\} := \nu(\tauzero)$.

We recall the formulas that make the three approximation layers (i), (ii), (iii), described in Section~\ref{sec:extrapolation:principle} more explicit.
The population approximation of the tail of the distribution function $F_X$ through the GPD from~\eqref{eq:tail:approximation:tilde:1} is 
\begin{align*}
	\tilde F_X(x) &= 1 - t_0^{-1}\left\{ 1 - H_{\gamma_X}\left(\frac{x - u_X}{a_X}\right) \right\}, \\
	& \stackrel{\gamma_X > 0}{=} 1 - t_0^{-1} \left(\frac{x-u_X}{u_X}\right)^{-1/\gamma_X} , \qquad x > u_X, \\
	& \stackrel{\gamma_X = 0}{=} 1 - t_0^{-1} \exp \left(-\frac{x - u_X}{a_X}\right), \qquad x > u_X .
\end{align*}
 Here, $H_{\gamma_X}$ is the distribution function of a GPD with shape parameter $\gamma_X \geq 0$ defined in \eqref{eq:domain:attraction:1}. 
Note that in the case $\gamma_X > 0$ we can choose the function $a_X(t) = \gamma_X u_X(t)$ \citep[][Theorem B.2.2]{dehaan:ferreira:2006}, which leads to the simplification above.
For $x \leq u_X$, we simply set $\tilde F_X(x) = F_X(x)$, meaning no tail approximation is used.

For a predictor value $x \geq u_X$, its transformation to Laplace margins $x^*$ and its corresponding approximation $\tilde x^*$ are given, respectively, by
\begin{align}\label{eq:trans:scale}
	x^* &:= Q_L \circ F_X(x) = -\log 2 + \log\{1/(1-F_X(x)) \},  \\
	\tilde x^* &:= Q_L \circ \tilde F_X(x) = -\log 2 + \log t_0 - \log \left\{ 1 - H_{\gamma_X}\left(\frac{x - u_X}{a_X}\right) \right\}, \nonumber \\
	&  \stackrel{\gamma_X > 0}{=} -\log 2 + \log\left\{ t_0 \left(\frac{x-u_X}{u_X}\right)^{1/\gamma_X}\right\}, \nonumber \\
	& \stackrel{\gamma_X = 0}{=} -\log 2 + \log t_0 + \frac{x - u_X}{a_X}, \nonumber
\end{align}
where $Q_L$ is the quantile function of the Laplace distribution and $\tilde F_X$ is the one previously computed.
Recall from \eqref{eq:notation} that  $U_Y(t) = Q_Y(1-1/t)$, thus  the GPD quantile approximation of the tail of $Y$ can be derived from \eqref{eq:quantile:tail:approximation} and can then be written as 
\begin{align}\label{eq:approx:U}
	\tilde U_Y(t) &= u_Y + \frac{a_Y}{\gamma_Y} \left\{ \left(\frac{t}{t_0}\right)^{\gamma_Y} - 1\right\} \\
	&\stackrel{\gamma_Y > 0}{=} u_Y \left(\frac{t}{t_0}\right)^{\gamma_Y} \nonumber \\
	& \stackrel{\gamma_Y = 0}{=} u_Y + a_Y \log\left(\frac{t}{t_0}\right) \nonumber ,
\end{align}
for $t > t_0$, and  again the simplification in the case $\gamma_Y > 0$ follows as $a_Y(t) = \gamma_Y a_Y(t)$. For  $ t \leq t_0$ we set $\tilde U_X(t) = U_X(t)$.
Finally, we denote the true median function on Laplace scale and its approximation function by
\begin{align*}
	m^*(x^*) &= \text{median}( Y^*  \mid X^* = x^*), \qquad x^* \in \mathbb R,\\
	\tilde m^*(x^*) &= a x^* + (x^*)^\beta b, \qquad x^* \in\mathbb R,
\end{align*}
respectively,
where $x^* = Q_L \circ F_X(x)$, $X^* = Q_L\circ F_X(X)$ and $Y^* = Q_L\circ F_Y(Y)$.

We now show some auxiliary results that will be used in the proof.
The first observation is that the approximate median behaves asymptotically as its true counterpart
	\beam\label{eq:m:to:infinity} 
	\left| \frac{m^*(x^*)}{ \tilde m^*(x^*)}  - 1 \right|   =  \left|  \frac{m^*(x^*) - \tilde m^*(x^*)}{\tilde m^*(x^*)} \right|  = \left| \frac{r(x^*)}{\tilde m^*(x^*)} \right| = o(1),
	\eeam 
where the residual function $r$ is defined in~\eqref{convergence:median:1}, and our assumption is that $r(x^*) = O\{\log(x^*)\}$ in the worst case, and $\tilde m^*(x^*)$ behaves as $a x^*$ if $a\in(0,1]$ and as $(x^*)^\beta b$ if $ a = 0, b, \beta > 0$.

The approximation of the transformed predictor converges uniformly to its target quantity, that is,
\begin{align}\label{xstar}
	\lim_{u_X \to \infty } \sup_{ u_X  <  x < \nu(u_X) } |x^* - \tilde x^*|  &= \lim_{u_X \to \infty} \sup_{  u_X < x < \nu(u_X) } \log \left|  \frac{1 - \tilde F_X(x)}{1-F_X(x)}\right| = 0.
\end{align}
This follows directly from Lemma~\ref{lem:uniform:tail:approximation}.
Here and in the sequel we always assume that the extrapolation boundary $\nu(u_X)$ is chosen as in~\eqref{eq:extra:condition} or~\eqref{eq:extra:condition:rho1}, depending on whether first- or second-order assumptions are satisfied, respectively.
 This also implies that 
\begin{align}\label{mstar}
	\lim_{u_X \to \infty } \sup_{ u_X  <  x < \nu(u_X) } |\tilde m^*(x^*) - \tilde m^*(\tilde x^*)|  &=  \lim_{u_X \to \infty } \sup_{ u_X  <  x < \nu(u_X) } | a(x^* - \tilde x^*) + b\{(x^*)^\beta - (\tilde x^*)^\beta\}| = 0,
\end{align}
since the function $x^\beta$, $\beta <1$, is Lipschitz for $x$ bounded away from 0. For future reference, we note that this, together with~\eqref{eq:m:to:infinity} , implies that in all cases we have
\begin{align}\label{m_infty}
	\lim_{x^* \to\infty} m^*(x^*) = \lim_{x^* \to\infty} \tilde m^*( x^*) = \lim_{x^* \to\infty} \tilde m^*( \tilde x^*) = \infty.
\end{align}
Moreover, it follows that 
\begin{align}\label{mstar_ratio}
	\lim_{u_X \to \infty } \sup_{ u_X  <  x < \nu(u_X) } \left|\frac{\tilde m^*(x^*)}{\tilde m^*(\tilde x^*)} - 1\right|  = 0.
\end{align}
A direct consequence is that for $a > 0$ we have
\begin{align}\label{mtilde}
	&\lim_{u_X \to \infty} \sup_{ u_X  <  x < \nu(u_X) }  \left|\frac{\tilde m^*(\tilde x^*)}{a x^*} - 1 \right| \nonumber \\
	& =  \lim_{u_X \to \infty} \sup_{ u_X  <  x < \nu(u_X) }   \left|\frac{(a x^* + (x^*)^\beta b) \tilde m^*(\tilde x^*)/ \tilde m^*( x^*)  - a x^*}{a x^*} \right| = 0,
\end{align}
and if $a = 0$ and $b,\beta > 0$ then
\begin{align}\label{mtilde:2}
	&\lim_{u_X \to \infty} \sup_{ u_X  <  x < \nu(u_X) }  \left|\frac{\tilde m^*(\tilde x^*)}{(x^*)^\beta b  } - 1 \right| = 0.
\end{align}

With the notation introduced at the beginning of this section, the regression function can be written as
\begin{align*} 
	m(x) &= \text{median}( Y \mid X = x) \\ 
	&= Q_Y^{-1} \circ F_L( m^*(x^*) )\\
	& = U_Y(2\exp\{m^*(x^*)\})\\
	& = U_Y(\tzero z),
\end{align*}
where  
\beam \label{eq:z:def}
z= 1/\big\{t_0 \big(1-F_L( m^*( x^*))\big)\big\} = 2\exp\{m^*(x^*)\}/t_0.
\eeam 
 On the other hand, the approximation $\tilde m(x)$ can be written as
\begin{align}\label{eq:ftilde:rewritting}
	\tilde m(x) &= \tilde Q_Y \circ F_L(\tilde m^*(\tilde x^*)) = \tilde U_Y(t_0 \tilde z),
\end{align}
where 
\beam \label{eq:z:extrapolation}
\tilde z = 1/\big\{t_0 \big(1-F_L( \tilde m^*( \tilde x^*))\big)\big\} = 2 \exp\{ \tilde m^*(\tilde x^*)\}/t_0.
\eeam
We note that $\tilde z = \tilde z(x)$ depends on the input predictor $x\in\mathbb R$ and would like to establish a uniform upper bound over the extrapolation interval $x\in (u_X,\nu(u_X))$.
 Since either $a=1$ and $\beta = 0$, or $a<1$, we have 
\begin{align}\label{bound:z} 
	\tilde z(x)t_0 /2 & = \exp\{ \tilde m^*( x^*)\}\exp\{ \tilde m^*(\tilde x^*) - \tilde m^*( x^*)\}\\
	& = \exp\{ a x^* + (x^*)^\beta b\}\exp\{ o(1)\} \\
	&\leq C \exp\{ x^*\}\exp\{ o(1)\} \\
	& = C \exp\{ Q_L\circ F_X(x) \}\exp\{ o(1)\} 	
\end{align}
for some constant $C>0$ and using the definition of $x^*$; the $o(1)$ is uniform on $(u_X, \nu(u_X))$ by~\eqref{mstar}. Consequently, by the fact that $\bar F_X(u_X) = 1/t_0$ and the form of the Laplace distribution function $F_L$, we obtain a uniform upper bound
\begin{align}\label{bound:z_unif} 
	\sup_{u_X < x < \nu(u_X)} \tilde z(x) = \sup_{u_X < x < \nu(u_X)} \frac{\bar F_X(u_X)}{\bar F_X(x)}\exp\{o(1)\} = O\left(\frac{\bar F_X(u_X)}{\bar F_X(\nu(u_X))} \right).	
\end{align}
This bound can be controlled by the assumptions~\eqref{eq:extra:condition} and~\eqref{eq:extra:condition:rho1} for $F_X$ on the extrapolation limit $\nu(u_X)$ under first- and second-order conditions, respectively.
Moreover, relying on the definitions of $z$ and $\tilde z $ in \eqref{eq:z:def} and \eqref{eq:z:extrapolation}, respectively,  we have
\begin{align}\label{eq:diff:medians}
	|\tilde z/z| = |\exp\{m^*(x^*) - \tilde m^*(\tilde x^*)\}| =
	|\exp\{r(x^*) + \tilde m^*(x^*) - \tilde m^*(\tilde x^*)\}| =
	O(\exp\{r(x^*)\}),
\end{align}
where the last relation is uniform on $(u_X,\nu(u_X))$ and holds because of~\eqref{mstar}.

We now study the relative error term, which can be written as 
\begin{align}\label{eq:rewritting:relative:error}
	\frac{m(x) - \tilde m(x)}{ m(x)}& = \frac{U_Y(\tzero z) - \tilde U_Y(\tzero \tilde z)}{  U_Y(\tzero   z)}.
\end{align}
The rest of the proof distinguishes two cases: $\gamma_Y =0$ and $\gamma_Y >0$.

\subsection*{Case $\gamma_Y = 0$}

\subsubsection{Intermediate result}

We start the proof by showing the following relation holds 
\begin{align}\label{eq:no:approx:U}
\lim_{u_X \to \infty }\sup_{ u_X  <  x < \nu(u_X) }  \left| \frac{U_Y(t_0 z) - U_Y(t_0 \tilde z)}{U_Y(t_0 z )}\right| = 0,
\end{align}
both under first- and second-order assumptions. In both cases, the proof is the same, with the only difference that for~\eqref{eq:diff:medians} to hold uniformly on $(u_X,\nu(u_X))$, we need to use the respective  extrapolation limits $\nu(u_X)$ in~\eqref{eq:extra:condition} and~\eqref{eq:extra:condition:rho1} for $F_X$, respectively. In particular, assumption~\eqref{eq:extra:condition:rho1:thm} is not needed for this part.

We consider the cases $a > 0$ and $a=0$ separately. 
\subsubsection*{Case $a > 0$} 

Let us first assume that $r(x^*) = O\{\log(x^*)\}$ and assumption~\eqref{ass_aU} holds.

Two scenarios are possible, either $\tilde z/ z > 1$, or $ z/  \tilde z \geq 1 $. 
We start with $\tilde z/z > 1$ and apply Proposition~\ref{lema:potters:bounds}. 
Set $t = t_0  z$ and note that $t\to \infty$ as $u_X \to \infty$ 
by the definition of $z$ in \eqref{eq:z:def} and the fact that $m^*(x^*) \to \infty$ 
as $x^* \to \infty$ from~\eqref{m_infty}.
Let $a_Y$ and $a_{0Y}$ be the auxiliary functions in there. Since $a_Y(t)/a_{0Y}(t) \to 1$, as $t \to \infty$, we can assume without losing generality that $a_Y(t) = a_{0Y}(t)$ for all $t\in\mathbb R$; this will not make a difference in any of the steps of the proof.
Then, we have that for any $\epsilon, \eta > 0$ there exists $t_1 = t_1(\epsilon,\eta)$ such that for $t\geq t_1$
\begin{align*}
	\left| \frac{U_Y(t) - U_Y(t \tilde z/z )}{ U_Y(t)} \right| &=   
	\left| \frac{U_Y(t) - U_Y(t \tilde z/z )}{ a_{Y}(t)} 
	\frac{a_{Y} (t)}{U_Y(t)}
	\right|  \\
	&= 	\left| \left( \frac{U_Y(t) - U_Y(t \tilde z/z )}{ a_{0Y}(t)} - \log(\tilde z /z )
	\right) \frac{a_{Y} (t)}{U_Y(t)}  + \log(\tilde z /z ) \frac{a_{Y} (t)}{U_Y(t)}
	\right| \\
	&\leq \epsilon \max\{(\tilde z / z)^\eta, (\tilde z / z)^{-\eta}\} \frac{a_{Y} (t)}{U_Y(t)} + \log(\tilde z/ z)\frac{a_{Y} (t)}{U_Y(t)}\\
	&\leq   C (\tilde z / z)^\eta \frac{a_{Y} (t)}{U_Y(t)} , 
\end{align*}
for some constant $C>0$.
From~\eqref{eq:diff:medians}, and if assumption that $r(x^*) = O\{\log(x^*)\}$ holds together with~\eqref{ass_aU}, it follows that  $|\tilde z/z| = O(x^*)$ uniformly on $(u_X, \nu(u_X))$.
Therefore, for $u_X$ sufficiently large, there exists $c > 0$ with
\begin{align}\label{eq:argument}
	\left|  \frac{U_Y(\tzero z) - U_Y(\tzero \tilde z)}{ U_Y(\tzero z)} \right| 
	& \leq   c\, (x^*)^\eta  \frac{a_Y(t_0 z)}{U_Y(\tzero z)}. \nonumber \\
	& =  c\, \frac{(x^*)^\eta }{ (m^*(x^*))^{\eta}}  \times (m^*(x^*))^{\eta}\frac{a_Y(2\exp\{m^*(x^*)\})}{U_Y(2\exp\{m^*(x^*)\})}.
\end{align}
If $a > 0$, then~\eqref{eq:m:to:infinity} implies that $m^*(x^*)/x^* \to a$, as $x^* \to \infty$, and thus $(x^*)^\eta / (m^*(x^*))^{\eta}  \to a^{-\eta}$, as $x^* \to \infty$. 
Moreover, since $\eta$ can be taken arbitrarily, we choose it to be $\eta = \delta$ with $\delta > 0$ as in assumption \eqref{ass_aU}. We then have
\begin{align*}
	\lim_{t \to \infty } \{ \log(t)\}^{\eta} \frac{a_Y(t)}{U_Y(t)} = \lim_{u \to \infty} \left\{-\log \bar F_Y(u)  \right\}^\eta \frac{\sigma_Y(u)}{ u } = 0,
\end{align*}
where the last equality follows by rewriting the functions $a_Y, U_Y$ with the notation from \eqref{eq:notation}, and the fact that $U_Y(t)
\to \infty$, as $t \to \infty$, since $Y$ has infinite upper end point. 
Since for $u_X\to\infty$, for any $x\in(u_X, \nu(u_X))$ we have that $m^*(x^*) \to\infty$, we conclude that also~\eqref{eq:argument} tends to zero uniformly on this interval. Therefore,
\begin{align}\label{eq:pe}
	\lim_{u_X \to \infty} \sup_{  \substack{ u_{X} < x  < \nu(u_X) \\
			\tilde z/z > 1
	} }   \left|  \frac{U_Y(\tzero z) - U_Y(\tzero \tilde z)}{ U_Y(\tzero z)} \right|  
	&= \lim_{u_X \to \infty}  ( \log(u_{X}) )^\eta \frac{a_Y( u_{X} )}{ U_Y( u_{X} )	} = 0.
\end{align}

Instead, if $z/\tilde z \geq 1 $, then letting $t = t_0  \tilde z$ in Proposition~\ref{lema:potters:bounds} and following similar steps as before we obtain
\beao 
\lim_{u_X \to \infty} \sup_{\substack{ u_{X} < x  < \nu(u_X) \\
		\tilde z/ z \leq  1
} } \left|  \frac{U_Y(\tzero z) - U_Y(\tzero \tilde z)}{ U_Y(\tzero  \tilde z)} \right| 
= \lim_{u_X \to \infty} (\log(u_X))^\eta \frac{a_Y(u_X)}{U_Y(u_X)} = 0,
\eeao 
where here we rely on the fact that $\tilde m^*(\tilde x^*)/x^* \to a$, as $x^* \to \infty$ and thus $\tilde m^*( \tilde x^*) \to \infty$, as $x^* \to \infty$, by~\eqref{mtilde}, to obtain the previous bound. 
This last relation also implies 
\begin{align}\label{eq:correct:renorm}
	\lim_{u_X \to \infty} \sup_{  \substack{ u_{X} < x  < \nu(u_X) \\
			\tilde z/z \leq 1
	} }    \left| \frac{U_Y(\tzero  z)}{ U_Y(\tzero \tilde z)}  - 1\right| 
	&= 0 ,
\end{align}
which then yields 
\begin{align}\label{eq:leq}
	\lim_{u_X \to \infty}\sup_{\substack{ u_{X} < x  < \nu(u_X) \\
			\tilde z/ z \leq  1
	} } \left|  \frac{U_Y(\tzero z) - U_Y(\tzero \tilde z)}{ U_Y(\tzero   z)} \right| = \lim_{u_X \to \infty} \sup_{  \substack{ u_{X} < x  < \nu(u_X) \\
	\tilde z/z \leq 1
} }    \left| \frac{U_Y(\tzero  \tilde z)}{ U_Y(\tzero z)}  - 1\right| =  0. \end{align} 
Finally, the~\eqref{eq:pe} and \eqref{eq:leq} together  imply \eqref{eq:no:approx:U}, which yields the desired result. 

If $r(x^*) = o(1)$, we do not need assumption~\eqref{ass_aU} to conclude~\eqref{eq:pe} and~\eqref{eq:leq}.

\subsubsection*{Case $a  = 0$}  
We now turn to the case $a=0$.  Let us first assume that $r(x^*) = O\{\log(x^*)\}$ and assumption~\eqref{ass_aU} holds.
The proof follows similar steps as the proof for $a >0 $.
From~\eqref{eq:argument} we have the bound 
\begin{align}\label{eq:argument:2}
	\left|  \frac{U_Y(\tzero z) - U_Y(\tzero \tilde z)}{ U_Y(\tzero z)} \right| 
	& \leq   c\, (x^*)^\eta   \frac{a_Y(t_0 z)}{U_Y(\tzero z)}. \nonumber \\
	& =  c\, \frac{(x^*)^\eta}{ (m^*(x^*))^{\eta/\beta}}  \times (m^*(x^*))^{\eta/\beta}\frac{a_Y(2\exp\{m^*(x^*)\})}{U_Y(2\exp\{m^*(x^*)\})}, 
\end{align}
For $a = 0$, by~\eqref{eq:m:to:infinity}, we have $m^*(x^*)/(x^*)^\beta \to b > 0 $, as $x^* \to \infty$, and this implies $(x^*)^{\eta} / (m^*(x^*))^{\eta/\beta}  \to b^{-\eta/\beta} < \infty$, as $x^* \to \infty$. 
Since $\eta$ can be taken arbitrarily, we may choose it as $\delta = \eta/\beta $ where $\delta > 0$ is as in \eqref{ass_aU}. Then, 
\begin{align*}
	\lim_{t \to \infty } \{ \log(t)\}^{\eta/\beta} \frac{a_Y(t)}{U_Y(t)} = \lim_{u \to \infty} \left\{ -\log  \bar F_Y(u) \right\}^{\delta} \frac{\sigma_Y(u)}{ u } = 0. 
\end{align*}
Finally, following similar steps as before imply 
\begin{align*}
	\lim_{u_X \to \infty} \sup_{  \substack{ u_{X} < x  < \nu(u_X) \\
			\tilde z/z > 1
	} }   \left|  \frac{U_Y(\tzero z) - U_Y(\tzero \tilde z)}{ U_Y(\tzero z)} \right|  
	&= 0 .
\end{align*}
The case $\tilde z/z \leq 1$ is analogous to the case $a > 0$; we omit the details. Putting everything together, under the assumptions of the Theorem~\ref{thm:main}, and notably the condition~\eqref{ass_aU}, we conclude that \eqref{eq:no:approx:U} holds.

Again, if $r(x^*) = o(1)$, then $|\tilde z/z| = O(1)$, and we do not need assumption~\eqref{ass_aU} to conclude~\eqref{eq:no:approx:U}.

\subsubsection{Bounds on the relative error in~\eqref{eq:rewritting:relative:error}}\label{sec:bounds_rel}
In the second step of the proof we come back to finding bounds for the relative error term in \eqref{eq:rewritting:relative:error}. 
Note that depending on how large is $x \geq u_X$, the approximated median $\tilde m^*(\tilde x^*)$ transformed back to the original scale, is above or below the threshold value $u_Y$.
This depends precisely on whether $\tilde z \leq 1$, or $\tilde z > 1$, respectively, with the notation from~\eqref{eq:z:extrapolation}.
We first consider the case $\tilde z \leq 1$.
In this case, we do not need to use the GPD extrapolation of $F_Y$ and we can express the approximation in \eqref{eq:ftilde:rewritting} as 
\begin{align*}
	\tilde m(x) &=  \tilde U_Y(\tzero \tilde z) = U_Y(\tzero \tilde z).
\end{align*}
We plug in this value in \eqref{eq:rewritting:relative:error} which yields
\begin{align*}
	\frac{m(x) - \tilde m(x)}{m(x)}& = \frac{U_Y(\tzero z) - U_Y(\tzero \tilde z)}{  U_Y(\tzero   z)}.
\end{align*}
Therefore, the relative error vanishes uniformly when $\tilde z \leq 1$ by~\eqref{eq:no:approx:U}.

It remains to consider the case where $\tilde z > 1$. 
We express the approximation as in \eqref{eq:ftilde:rewritting} by 
\begin{align*}
	\tilde m(x) &= \tilde U_Y(t_0 \tilde z) = U_Y(\tzero)  + a_Y \log \tilde z ,
\end{align*}
where the last equality follows by \eqref{eq:approx:U}.
We consider the relative error
\begin{align}
	\notag \frac{m(x) - \tilde m(x)}{\tilde m(x)}& = \frac{U_Y(\tzero z) - U_Y(\tzero)  - a_Y \log \tilde z }{ U_Y(\tzero)  + a_Y \log \tilde z }\\
		 \notag & = 	
	  \frac{U_Y(t_0 \tilde z)}{U_Y(\tzero)  + a_Y \log \tilde z }\times \left[ \frac{U_Y(\tzero z) - U_Y(\tzero \tilde z)  }{U_Y(t_0 \tilde z) } + 1\right]  - 1 \\
	\label{decomposition} & = (I) \times [ (II) +1] -1.
\end{align}
Since by \eqref{eq:no:approx:U}, $U_Y(t_0 z)/ U_Y(t_0 \tilde z)$ goes to one uniformly, the term $(II)$ goes to zero uniformly. Therefore, 
it suffices to show that term $(I)$ goes to 1 uniformly.

For all $\eta > 0$, for $t_0$ large enough, Proposition~\ref{lema:potters:bounds} implies the bound
\begin{align*}
	&\left|\frac{U_Y(t_0 \tilde z)}{ U_Y(t_0) + a_Y \log \tilde z}  - 1  \right| \leq \frac{a_Y (\tilde z)^\eta}{U_Y(t_0) + a_Y \log \tilde z} \leq  a_Y (\tilde z)^\eta/U_Y(t_0) 
\end{align*}
since $a_Y = a_Y(t_0)$ and $\log\tilde z$ are both positive. With the uniform bound on $\tilde z$ in~\eqref{bound:z_unif} we obtain
\begin{align*}
	\sup_{ u_X  <  x < \nu(u_X) } \left|\frac{U_Y(t_0 \tilde z)}{ U_Y(t_0) + a_Y \log \tilde z}  - 1  \right|  \leq \frac{a_Y}{U_Y(t_0)}O\left(\frac{\bar F_X(u_X)}{\bar F_X(\nu(u_X))} \right)^\eta \to 0, \quad t_0 \to\infty,	
\end{align*}
since $a_Y(t_0)/U_Y(t_0) \to 0$ as $t_0\to\infty$ (see Remark~\ref{rem:sigma_U}) 
and the $O(\cdot)$ term is bounded because of assumption~\eqref{eq:extra:condition} for $F_X$.
This establishes~\eqref{main_conv} and concludes the proof in the case $\gamma_Y = 0$ under first-order conditions.

Now assume that $X$ and $Y$ also satisfy the second-order condition \eqref{eq:second:order:intro} 
with second-order parameters $\rho_X,\rho_Y \leq 0$, respectively.
Let $A_Y$ be the second-order auxiliary function for $F_Y$ as defined in~\eqref{eq:notation}.
We detail the proof only for the case $\rho_Y = 0$; the case $\rho_Y < 0$ is similar. 
We have the same decomposition of the relative error as in~\eqref{decomposition}, and term $(II)$ converges uniformly to one under the assumption that $\nu(u_X)$ satisfies~\eqref{eq:extra:condition:rho1}. 
For term $(I)$ we can use the stronger result from Theorem~\ref{eq:thm:second:order:bounds} 
to obtain for $\eta > 0$ as in assumption~\eqref{eq:extra:condition:rho1:thm}, 
for $t_0$ large enough the bound
\begin{align*}
	&\left|\frac{U_Y(t_0 \tilde z)}{ U_Y(t_0) + a_Y \log \tilde z}  - 1  \right| \leq 
	\frac{ A_Y(t_0)  (\tilde z)^\eta a_Y\log \tilde z}{U_Y(t_0) + a_Y \log \tilde z} \leq  A_Y(t_0)(\tilde z)^\eta, 
\end{align*}
since $U_Y(t_0)$ is eventually positive. With the uniform bound on $\tilde z$ in~\eqref{bound:z_unif} we obtain
\begin{align}\label{eq:second:orderI}
	\sup_{ u_X  <  x < \nu(u_X) } \left|\frac{U_Y(t_0 \tilde z)}{ U_Y(t_0) + a_Y \log \tilde z}  - 1  \right|  \leq O\left(A_Y(t_0)\left\{\frac{\bar F_X(u_X)}{\bar F_X(\nu(u_X))}\right\}^\eta \right) \to 0, \quad t_0 \to\infty,	
\end{align}
because of assumption~\eqref{eq:extra:condition}, and therefore term $(I)$ in~\eqref{decomposition} converges uniformly to one.
This establishes~\eqref{main_conv} for $\gamma_Y = 0$ under the larger extrapolation 
boundaries for second-order conditions.

\subsection*{Case $\gamma_Y > 0$} 

From~\eqref{eq:diff:medians} and the assumption that $\log(x^*) = o(1)$ in the case $\gamma_Y > 0$ it follows that 
\beam \label{eq:z:gamma0} 
	\lim_{u_X \to \infty} \sup_{ u < x < \nu(u_X)} | z/ \tilde z | = 1,
\eeam 
both under first- and second-order conditions on $F_X$ with the respective assumptions on $\nu(u_X)$.
Since $U_Y$ is regularly varying with tail-index $\gamma_Y>0$, Potter's bound \cite[Proposition B.1.9]{dehaan:ferreira:2006} implies that for $\epsilon > 0$, there exists $ t_1 = t_1(\epsilon)$ such that, for all $t, tx > t_1$,
\begin{align}\label{eq:pott:bounds}
	(1-\epsilon) \max\{ x^{\gamma_Y + \epsilon}, x^{\gamma_Y - \epsilon}\} \leq  \frac{U_Y(tx)}{ U_Y(t)} \leq (1+\epsilon) \max\{ x^{\gamma_Y + \epsilon}, x^{\gamma_Y - \epsilon}\}.
\end{align}
Setting $t = t_0 z$ and $x = \tilde z/z$, and noting that $t_0 z \to \infty$ for all $x \in (u_X, \nu(u_X))$, as $u_X\to\infty$, because of~\eqref{eq:z:def} and~\eqref{m_infty}, we obtain   
\begin{align}\label{eq:bound:approx:med:gamma:positive}
	\lim_{u_X \to \infty}  \sup_{u < x < \nu(u_X)} \left| \frac{U_Y(t_0 z )}{ U_Y(t_0 \tilde z)} - 1 \right| = 0,
\end{align}
where we used~\eqref{eq:z:gamma0}.

Similarly to $\gamma_Y = 0$ in Section~\ref{sec:bounds_rel}, we first consider the case $\tilde z \leq 1$, in which case we have 
\begin{align*}
	\tilde m(x) = U_Y(\tzero \tilde z).
\end{align*}
The relative error then has the form 
\begin{align*}
	\frac{m(x) - \tilde m(x)}{ m(x)} = \frac{U_Y(\tzero z) - U_Y(\tzero \tilde z)}{ U_Y(\tzero z)},
\end{align*}
and we see by \eqref{eq:bound:approx:med:gamma:positive} that
\begin{align}\label{m_z_small}
	\lim_{u_X\to\infty} \sup_{ \substack{ u_X < x < \nu(u_X) \\ \tilde z \leq 1 } } \left| \frac{m(x) - \tilde m(x)}{ m(x)} \right| \leq  	\lim_{u_X\to\infty} \sup_{ \substack{ u_X < x < \nu(u_X) } } \left|  \frac{U_Y(\tzero z) - U_Y(\tzero \tilde z)}{ U_Y(\tzero z)} \right| = 0. 
\end{align}
We now consider the case where $\tilde z > 1$. 
From~\eqref{eq:approx:U} and~\eqref{eq:ftilde:rewritting} we express the approximation as 
\begin{align*}
	\notag\tilde m(x) &= \tilde U_Z(t_0 \tilde z )= U_Y(\tzero) \tilde z^{\gamma_Y},
\end{align*}
and thus we can rewrite the relative error as 
\begin{align}
	\notag \frac{m(x) - \tilde m(x)}{\tilde m(x)}& = \frac{U_Y(\tzero z)}{ U_Y(\tzero)\tilde z^{\gamma_Y}} - 1\\
	\notag & = \frac{U_Y(t_0 z) }{U_Y(t_0 \tilde z)} \times \frac{U_Y(\tzero \tilde z)}{ U_Y(\tzero)\tilde z^{\gamma_Y}}  - 1\\
	\label{error_decomp}& = (I) \times (II) -1,
\end{align}
Since $(I)$ converges to one uniformly as $u_X \to \infty$ by~\eqref{eq:bound:approx:med:gamma:positive}, we focus on term $(II)$.

Since $\tilde z > 1$, we obtain from Proposition~\ref{lema:potters:bounds} and the fact that $a_Y(t_0) = \gamma U_Y(t_0)$ in the case $\gamma_Y>0$, that for any $\epsilon > 0$, there exists $t_1$ such that for $t_0 \geq t_1$
\begin{align*} 
	 \left| \frac{U_Y(t_0 \tilde z)}{U_Y(t_0) \tilde z^{\gamma_Y}} - 1 \right|  \leq \epsilon \max\{(\tilde z)^{\epsilon}, (\tilde z)^{-\epsilon}\} 
\end{align*} 
Using~\eqref{bound:z_unif} and assumption~\eqref{eq:extra:condition} on the sequence $\nu(u_X)$, we have
\begin{align*} 
	\lim_{u_X \to \infty}\sup_{u_X < x < \nu(u_X) } 	\left| \frac{U_Y(t_0 \tilde z)}{U_Y(t_0) \tilde z^{\gamma_Y}} - 1 \right|  = \epsilon  \, O\left( \frac{\bar F_X(u_X)}{\bar F_X(\nu(u_X))}\right) \leq \epsilon M,
\end{align*} 
for some constant $M$. Since $\epsilon$ is arbitrary, we conclude that term $(II)$ goes to one uniformly and therefore
\beam\label{eq:res} 
\lim_{u_X \to \infty}	\sup_{ \substack{ u_X < x < \nu(u_X) \\ \tilde z > 1 } }  \left|\frac{m(x) - \tilde m(x)}{\tilde m(x)}\right| = 0,
\eeam 
and this together with~\eqref{m_z_small} shows that \eqref{main_conv} holds.

Now consider the case where $X$ and $Y$ satisfy the second-order condition \eqref{eq:second:order:intro} (or equivalently \eqref{eq:second:order:condition}) with a second order parameter $\rho_X,\rho_Y \leq 0$, respectively. 
We first note that if $\tilde z\leq 1$, then we can use the same arguments as under first-order conditions to obtain~\eqref{m_z_small}; note that for~\eqref{eq:z:gamma0} we require that $\nu(u_X)$ satisfies~\eqref{eq:extra:condition:rho1} for $F_X$. 
Suppose now that $\tilde z > 1$ and note that we have the same decomposition as in~\eqref{error_decomp}. The term $(I)$ does to one uniformly because of~\eqref{eq:bound:approx:med:gamma:positive}. We consider separately the cases $\rho_Y< 0$ and $\rho_Y = 0$. When $\rho_Y = 0$, then applying \eqref{eq:ftx:expansion} we see that for $\eta > 0$ arbitrarily close to zero, and since $\tilde z > 1$, 
\begin{align*}
\lim_{u_X \to \infty}	\sup_{\substack{ u_X < x < \nu(u_X) \\ \tilde z > 1 } } \left| \frac{U_Y(\tzero \tilde z)}{ U_Y(\tzero)\tilde z^{\gamma_Y}}  - 1\right| &= \lim_{u_X \to \infty} 	\sup_{\substack{ u_X < x < \nu(u_X) \\ \tilde z > 1 }  } O(A_Y(t_0) \tilde z^\eta ) \\
&= \lim_{u_X \to \infty}  O\left\{A_Y(t_0)\left( \frac{\bar F_X(u_X)}{\bar F_X(\nu(u_X))}\right)^\eta \right\}\\
& = 0.   
\end{align*}
Here we used the uniform bound~\eqref{bound:z_unif} and assumption~\eqref{eq:extra:condition:rho1:thm}. This means that term $(II)$ in~\eqref{error_decomp} converges uniformly to one. In the case $\rho_Y <  0$ we can use  \eqref{eq:ftx:expansion} for the case $\rho_Y < 0$ to obtain a similar result. Putting things together, this means that~\eqref{main_conv} holds also under the second-order condition framework.

\section{Proofs of Section~\ref{sub:sec:examples}}

\subsection{Proofs Section~\ref{sec:regularly:v:tails}}\label{proof:sec:asymptotic:dep:mod:1}

The proof of Proposition~\ref{example:asymptotic:dependence} follows from the following lemma.

\begin{lemma}\label{lem:regular:variation:fx}
	Suppose that $X$ satisfies the domain of attraction condition with shape parameter $\gamma_X > 0$, that is, $\bar F_X \in \text{RV}(-\alpha_X)$, where $\alpha_X  = 1 / \gamma_X$. Let $f$ be a continuous, regularly varying function with index $\alpha_f > 0$, that is, $f \in \text{RV}(\alpha_f)$.
	Let the noise term $\varepsilon$ not be more heavy-tailed than the signal $f(X)$, that is, there exists $c_\epsilon \geq 0$ such that  $\P(\varepsilon > x) \sim c_\epsilon \P(f(X)>x),$ as $x\to\infty$. Then the following holds:
	\begin{itemize}
		\item[(i)] $f(X)$ satisfies the domain of attraction condition with shape parameter $\gamma_{f(X)}$ satisfying $\alpha_{f(X)} = 1/\gamma_{f(X)} = \alpha_X/\alpha_f $;
		\item[(ii)] The distribution of $Y$ satisfies $\bar F_Y \sim (1+c_\epsilon)\bar F_{f(X)}$.
		\item[(iii)] The median on the transformed scale satisfies
		\begin{align}\label{eq:median:approx}
			| \M( Y^* \mid X^* = x^*) - x^* - \log(1/(1+c_\epsilon) | \to 0, \quad x^*\to \infty.
		\end{align}
	\end{itemize}
\end{lemma}

\begin{proof}[Proof Proposition~\ref{example:asymptotic:dependence}]
	By the assumptions of Proposition~\ref{example:asymptotic:dependence}, also the assumptions of Lemma~\ref{lem:regular:variation:fx} are satisfied. Therefore, $Y$ is in the domain of attraction with shape parameter $\gamma_Y >0 $ by (ii).
	By equation~\eqref{eq:median:approx} in (iii), we conclude that all assumptions of Theorem~\ref{thm:main} are verified for this model. 
\end{proof}

\begin{proof}[Proof Lemma~\ref{lem:regular:variation:fx}]
	We start by showing that the random variable $f(X)$ is regularly varying and $\bar F_{f(X)} \in RV(-\alpha_{f(X)})$ where $\alpha_{f(X)} = \alpha_X/\alpha_f$. 
	To see this, notice that by \citet[Proposition 2.6(vii)]{resnick:2007}, there exists an absolutely continuous, strictly monotone function $\bar f \in \text{RV}(\alpha_f)$ with $\lim_{x\to\infty} f(x) / \bar f(x) = 1$.
	Moreover, $\bar f^{-1} \in \text{RV}(1/\alpha_f)$ by \citet[Proposition 2.6(v)]{resnick:2007}. Since we assume in Section~\ref{sec:add_noise} that $X$ is bounded from below, and since $f$ is continuous and $f(x) \to \infty$, as $x \to \infty$, the fact that $f(X)>x$ implies that $X>t(x)$ almost surely for the sequence $t(x) = \inf\{y\geq 0: f(y) \geq x + \kappa\} \to \infty$ for some $\kappa>0$. Consequently, for $\delta > 0$, we have for $x$ large enough that ${f(X)}/{\bar f(X)} \leq 1 + \delta$ almost surely on the event $f(X)>x$, and thus 
	\begin{align}\label{eq:approx:fbar}
		\P(f(X) > x) = \P(\bar f(X) \frac{f(X)}{\bar f(X)} > x) \leq  \P(\bar f(X) > x/(1+\delta)) = \bar F_X \circ \bar f^{-1}(x/(1+\delta)).
	\end{align}
	Similarly we get a lower bound by using $1 -\delta$. By \citet[Proposition 2.6(v)]{resnick:2007}, $\bar F_X \circ \bar f^{-1}$ is a regularly varying function at infinity with index $- \alpha_X/\alpha_f$. We therefore obtain
	\begin{align}\label{eq:fX}
		\lim_{x\to\infty} \frac{\P(f(X) > tx)}{\P(f(X) > x)} \leq  
		\lim_{x\to\infty} \frac{\bar F_X \circ \bar f^{-1}(tx/(1+\delta))}{\bar F_X \circ \bar f^{-1}(x/(1-\delta))} = t^{- \alpha_X/\alpha_f} \left(\frac{1+\delta}{1-\delta}\right)^{\alpha_X/\alpha_f}.
	\end{align}
	Since $\delta>0$ was arbitrary, we obtain that $\bar F_{f(X)} \in \text{RV}(-\alpha_X/\alpha_f)$.
	Furthermore, applying again Equation~\eqref{eq:approx:fbar}, for all $\delta > 0$, we have
	\begin{align*}
		\lim_{x \to \infty}\frac{ \bar F_{f(X)} }{ \bar F_X \circ \bar f^{-1}(x) } \leq
		\lim_{x \to \infty}
		 \frac{\bar F_X \circ \bar f^{-1}(x/(1+\delta))}{ \bar F_X \circ \bar f^{-1}(x) } 
		 = (1+\delta)^{\alpha_X/\alpha_f}.
	\end{align*}	
	It is also possible to obtain a lower bound by $(1-\delta)^{\alpha_X/\alpha_f}$, and again since $\delta > 0$ is arbitrary we obtain $\bar F_{f(X)} \sim \bar F_X \circ \bar f^{-1} $. 
	\par 
	We now compute the tail asymptotics of the distribution of $Y$. 
	For this, 
	note that by \citet[Lemma 3.1(3)]{mikosch2014modeling} the random variable $f(X) + \varepsilon$ is regularly varying with index $-\alpha_{f(X)}$ since it is the sum of two regularly varying random variables. 
	Moreover, by Feller's convolution Lemma \cite[Lemma 1.3.1]{mikosch2014modeling}
	\begin{align*}
		\P( Y > x )  = \P(f(X) + \varepsilon > x) &\sim \P(f(X) > x ) + \P(\varepsilon > x)\\
		&\sim (1+c_\epsilon)\P(f(X) > x )\\
		& \sim (1+c_\epsilon) \bar F_X \circ \bar f^{-1} (x), \quad x\to \infty,
	\end{align*} 	
	where the last tail equivalences follow since $\P(\varepsilon > x) \sim c_\epsilon \P(f(X) > x)$, as $x \to \infty$, and since we have shown
	previously that $\bar F_{f(X)} \sim \bar F_X \circ \bar f^{-1} $. 
	Finally, since $f$ is regularly varying we can argue again as in \eqref{eq:fX} to conclude 
	\beao 
	\bar F_{Y} \circ f(x) \sim (1+ c_\epsilon) \bar F_X(x), \quad x \to \infty. 
	\eeao 
	Finally, we conclude by Lemma~\ref{lem:median:equation} that
	 \beao 
	\lim_{x^* \to \infty }|\M[Y^* | X^* = x^*] - x^* - \log(1/(1+c_\epsilon))|  = 0,
	\eeao 
	which gives the desired result to conclude the proof of Lemma~\ref{lem:regular:variation:fx}. 
	 \end{proof}

\subsection{Proofs Section~\ref{sec:gumbel:subexponential}}\label{proof:sec:asymptotic:dep:mod:2}

\begin{proof}[Proof  Lemma~\ref{ex:exp:chi:subexponential}]
	
	We show that there exists $\rho < 1$ and $s \geq 0$ such that 
	 \begin{itemize}
		\item[(i)] $\chi(yt)/\chi(t) \leq y^\rho, \quad \text{for all } y \geq 1, t \geq s$;
		\item[(ii)] the function $\exp\{-(2-2^\rho) \chi(t)\}$ is integrable on $[2,\infty)$.
	 \end{itemize}
	Subexponentiality of $F$ then follows by \citet[][Theorem 3]{cline:1986}. 
	
	To show (i), note that by the Potter bounds \citep[][Proposition B.1.9]{dehaan:ferreira:2006}, for $\epsilon > 0$ with $0 <\alpha_\chi  - \epsilon  <\alpha_\chi + \epsilon < 1$, there exists $s > 0$ such that for all $t, ty \geq s$,
	\[
		 (1-\epsilon) y^{\alpha_\chi - \epsilon} \leq {\chi(yt)}/{\chi(t)} \leq (1+\epsilon) y^{\alpha_\chi + \epsilon}.
	\]
	Thus, there exists $\epsilon^\prime$ with $\rho := \alpha_\chi + \epsilon^\prime < 1$ and  ${\chi(yt)}/{\chi(t)} \leq y^{\rho}$, showing (i).

	We now show (ii) holds. Note that Potter bounds hold uniformly for all $t \geq s$ and $ty \geq s$. This implies the existence of a constant $C = C(s) > 0$ such that for $t \geq s$, we have $\chi(t) \geq C y^{\alpha_X - \epsilon}$. Then, for all $\rho \in (0,1)$ we have 
	\[
		\exp\{-(2-2^\rho)\chi(t)\} \leq \exp\{-\chi(t)\} \leq \exp\{- C t^{\alpha_\chi - \epsilon}\}, \quad t \geq s.
	\]
	Finally, since $\alpha_\chi -\epsilon > 0$, the last term is integrable on $[2,\infty)$ and this concludes the proof of~(ii).
	Moreover, if $\chi$ verifies the von Mises condition ${\chi^{\prime \prime} (t) }/{({\chi^\prime}(t))^2} \to 0$ as $t \to \infty$, then it follows by \citet[][Example 3.3.24]{embrechts:kluppelberg:mikosch:2013}
	that $F$ satisfies the domain of attraction condition with shape parameter $\gamma = 0$.
	
\end{proof}

\begin{lemma}\label{lem:model2:subexpo}
	Suppose that $X$ admits a distribution $F_X$ such that $F_X = 1-\exp\{ \chi(x)\}$, where $\chi$ is a strictly monotone regularly varying function with index $\alpha_\chi \in(0,1)$.
	Suppose $\varepsilon$ is independent of $X$ and admits a subexponential distribution.
	Let $f$ be a strictly monotone regularly varying function $f \in\text{RV}(\alpha_f)$ with tail index $\alpha_f  > 0$.
	Further suppose 
	\[\alpha_\chi / \alpha_f < 1,\]
	and that  $\P(\varepsilon > x) \sim c_\epsilon \P\{f(X)>x\},$ for $c_\epsilon \in [0,\infty)$, as $x\to\infty$.
	Then we have $\bar F_Y \sim (1+c_\epsilon) \bar F_{f(X)}$ and 
	\begin{align*}
		|\M[Y^* | X^* = x^*] - x^* - \log(1/(1+c_\epsilon))| \to 0, \quad x^* \to\infty.
	\end{align*} 
\end{lemma}

Relying on Lemma~\ref{lem:model2:subexpo} we are ready to show Proposition~\ref{example:asymptotic:dependence2} holds.
\begin{proof}[Proof Proposition~\ref{example:asymptotic:dependence2}]

By Lemma~\ref{ex:exp:chi:subexponential}, $X$ satisfies the the domain of attraction condition with $\gamma_X = 0$.
We start by showing that $Y$ is in the domain of attraction with $\gamma_Y = 0$. 
Relying on Lemma~\ref{lem:model2:subexpo} we have that $\bar F_Y \sim (1+c_\epsilon) \bar F_{f(X)}$. 
Hence, it is enough to check that $f(X)$ is in the domain of attraction with $\gamma_{f(X)} = 0$.
For this purpose, let $g(x) = \chi \circ f^{-1}(x)$, and notice we assumed that $g$ satisfies the second-order von Mises condition: $-{g^{\prime \prime} (x) }/{({g^\prime}(x))^2} \to 0$, as $x \to \infty$.
By~\citet[][Example 3.3.24]{embrechts:kluppelberg:mikosch:2013}, this condition implies $f(X)$ is in the domain of attraction with $\gamma_{f(X)} = 0$. 
This also implies that $Y$ satisfies the first-order condition with $\gamma_Y = 0$.
Together with Lemma~\ref{lem:model2:subexpo}, we have thus verified all the assumptions of Theorem~\ref{thm:main}.
\end{proof}

\begin{proof}[Proof Lemma~\ref{lem:model2:subexpo}]
	We start by showing that the distribution of $f(X)$ with distribution function $F_{f(X)}$ is a subexponential. Indeed, since $f$ is a strictly 
	monotone function, it admits an inverse $ f^{-1} \in \text{RV}(1/\alpha_f)$ by \cite[Proposition 2.6(v)]{resnick:2007}, such that
	\beam\label{eq:bound:margins:subexponential} 
	F_{f(X)}(x) = \P( f(X) \leq x ) = \P( X \leq f^{-1}(x)) = 1-\exp\{ - \chi \circ f^{-1} (x) \}.
	\eeam 
	Hence, applying \cite[Proposition 2.6(v)]{resnick:2007}, we see that $\chi \circ f^{-1}$ is regularly varying with index $\alpha_\chi/\alpha_f < 1$. 
	Lemma~\ref{ex:exp:chi:subexponential} then shows that $f(X)$ also has a subexponential distribution. 
	This fact together with the assumption that $\P(\varepsilon > x) \sim c_\epsilon \P(f(X) > x),$ as $x \to \infty$   implies \begin{align*}
		\bar F_Y(x) = \P(f(X) + \varepsilon > x) \sim (1+c_\epsilon) \P(f(X) > x) = (1+c_\epsilon) \bar F_X\circ  f^{-1}(x),
	\end{align*}
	where the first equivalence follows from \cite[Theorem 1]{cline:1986}. 
	Thus, since $f$ is a strictly monotone we have 
	\beao   \bar F_Y \circ f \sim (1+c_\epsilon) \bar F_X, \quad  x\to\infty. 
	\eeao 
	Finally, we can apply Lemma~\ref{lem:median:equation} to conclude
	\beao 
	\lim_{x^* \to \infty }|\M[Y^* | X^* = x^*] - x^* - \log(1/(1+c_\epsilon))|  = 0. 
	\eeao 
	
\end{proof}

\subsection{Proof Section~\ref{subsec:light:tails:gumbel}}\label{sec:model:3}

This section builds on the theory developed in \cite{asmussen:hashorva:laub:taimre:2017} together with the original contributions 
from \cite{bal1993}.
We consider random variables $X$, $\varepsilon$ with distribution $F_X$ and $F_\varepsilon$ admitting densities $g_X,g_\epsilon$ such that
\begin{align}\label{eq:den:assumptions}
	g_X(x) &\sim d_X \, x^{\eta_X + \tau_X - 1} \exp\{ - c_X x^{\tau_X}\}, \quad x \to \infty \nonumber \\
	g\varepsilon(x) &\sim d_\varepsilon \, x^{\eta_\varepsilon + \tau_\eps  - 1} \exp\{ - c_\varepsilon  x^{\tau_\eps  }\},  \quad x \to \infty,
\end{align}
with $\tau_X, \tau_\eps > 1$, $\eta_X, \eta_\varepsilon \in\mathbb R$ and $c_X, c_\varepsilon, d_X, d_\varepsilon > 0$. 
To prove Proposition~\ref{prop:weibull:light:tails} we need the lemma below whose proof is given at the end on this section.

\begin{lemma}\label{lem:weibull:light:tails:2}
Let $X$, $\varepsilon$ be as in \eqref{eq:den:assumptions:2},
$f(x) = x^{\alpha_f}$ for $x > 0$ and $Y = f(X) + \varepsilon$. Assume 
\[
\tau_X/\alpha_f = \tau_\eps > 1.  
\]
Then, $\bar F_Y(x) \sim d x^\eta \exp\{ - c x^{\tau_\eps}\}$,  for suitable $d, \eta >0$, with $x= ac_X$ and 
\begin{align}\label{eq:median:result}
		|\M[Y^* | X^* = x^*]  	- a x^*| = O\{ \log(x^*)\}, \quad x^* \to\infty.
	\end{align} 
	with $a\in (0,1)$ as in Proposition~\ref{prop:weibull:light:tails}.
\end{lemma}

\begin{proof}[Proof Proposition~\ref{prop:weibull:light:tails}]
Under the assumptions of Proposition~\ref{prop:weibull:light:tails}, we can use Lemma~\ref{lem:weibull:light:tails:2}. To verify the assumptions of Theorem~\ref{thm:main} under first-order conditions, it thus sufficies to show that $Y$ satisfies the domain of attraction condition~~\eqref{eq:domain:attraction:1} with $\gamma_Y = 0$, and that~\eqref{ass_aU} holds.
First note that with $\chi(x) =  c x^{\tau_\eps} - \eta \log(x) + \log(d)$,
we can write $F_Y(x) \sim  \exp\{- \chi(x)\}$ as $x  \to \infty$. It is straightforward to check that $\chi$ verifies the von Mises condition $-{\chi^{\prime \prime} (x) }/{({\chi^\prime}(x))^2} \to 0$ as $x \to \infty$. \citet[][Example 3.3.24]{embrechts:kluppelberg:mikosch:2013} then implies that $Y$ satisfies the domain of attraction  condition with $\gamma_Y = 0$. 

Similar calculations as in \eqref{eq:qx:weibull} below show that $U_Y(t) = \{ \log(t/c) + O(\log \log t)\}^{1/{\tau_\eps}}$, and thus~\eqref{eq:first:order:U} with scaling function $ \sigma_Y\{U_Y(t)\} = a_Y(t) = {\tau_\eps}^{-1} (\log t)^{1/{\tau_\eps} -1}$. This implies that for any $\delta \in (0,1)$ we have
	 \begin{align*}
		\lim_{u \to \infty}	\frac{\sigma_Y(u) }{u}
		\left\{- \log \bar F_Y(u) \right\}^{\delta}  &= 		
		\lim_{t \to \infty}	\frac{a_Y(t)}{U_Y(t)}
		\left(\log t \right)^{\delta} \\
		&= 	\lim_{x \to \infty}	\frac{a_Y\{\exp(x)\}x^{\delta}}{U_Y\{\exp(x)\}}
		  \\
		 & = \lim_{x\to\infty} O( x^{\delta -1}) = 0.
	 \end{align*}
 	 This verifies \eqref{ass_aU} from Theorem~\ref{thm:main}.
\end{proof}

\begin{proof}[Proof Lemma~\ref{lem:weibull:light:tails:2}] 
Under assumptions~\eqref{eq:den:assumptions} we can compute the asymptotic development of the density of $f(X)$ as
\beam \label{eq:comp:density:fx} 
g_{f(X)}(x) \sim d_X x^{ (\eta_X + \tau_X - 1)/\alpha_f } \exp\{ - c_X  x^{\tau_X/\alpha_f}\}, \quad x\to \infty.
\eeam  
By \citet[][Theorem 3.1]{asmussen:hashorva:laub:taimre:2017} we have
 \begin{align*}
	\bar F_Y(x) = \P(f(X) + \varepsilon > x) \sim k x^{\gamma} \exp\{ - c x^{\tau}\},
\end{align*}
where $\gamma, k >0$ and $c = a c_X$. We next compute the transformation from the Laplace margins to the predictor distribution. 
From the density expression in \eqref{eq:den:assumptions} we now derive bounds on the quantile function $Q_X$ of $X$. \citet[Equation 3.3]{asmussen:hashorva:laub:taimre:2017} yields 
\[
 \bar F_X(x) \sim \frac{d_X}{\tau c_X} x^{\eta_X}\exp\{-c_X x^{\tau_X}\} =:h(x), \quad x\to\infty.
\] 
Denoting $v = v(x)$ the solution to the equation $ h(v) = \bar F_L(x) = \exp\{-x\}/2$, for large $x>0$, we can write this exactly in terms of the Lambert function $W$ as 
\[
 	v = \left[  \frac{\eta_X}{c_X \tau_X } 
 		W\left( c_X \tau_X \eta_X^{-1}
 		\exp\{ \tau_X \eta_X^{-1} x  - q \tau_X \eta_X^{-1} \} \right) \right]^{1/\tau_X},
\]
where $q = d_X/\tau c_X$ and the Lambert function $W$ is given by
\beam \label{eq:inv:gaussian}
W(x) = \log x - \log \log x + O\Big(  \frac{ \log \log x}{ 2 \log x } \Big), \quad x \geq e.
\eeam 
A series expansion shows
\[
\left| W\left( c_X \tau_X \eta_X^{-1}
 		\exp\{ \tau_X \eta_X^{-1} x  - q \tau_X \eta_X^{-1} \} \right) - \tau_X \eta_X^{-1} x \right| = O\left\{\log(x)\right\}, \quad x\to\infty
\]
 and since then $Q_X \circ F_L(x) = v(x)$, we have
\beam \label{eq:qx:weibull2}
|Q_X \circ F_L(x)  - (x/c_X)^{1/\tau_X}| = O\left(\log(x) x^{1/\tau_X - 1} \right), \quad x \to \infty.
\eeam 
Moreover, by~\eqref{eq:bound:margins:subexponential}, we can compute an approximation of the transformation of the response distribution to Laplace margins: $ Q_L \circ F_Y(x) =  -\log \bar F_Y(x) - \log 2$, which yields
\beam \label{eq:fy:weibull2} 
|Q_L \circ F_Y(x) - c x^{\tau_\varepsilon} | = O\{\log(x)\}, \quad x \to \infty.  
\eeam 
In particular, this implies there exists $C > 0$, such that we can now bound the median as 
\begin{align*}
	\M[Y^* | X^* = x^*]  &= Q_L \circ F_Y( \M[Y | X = Q_X\circ F_L(x^*)] ) \\
		&= Q_L \circ F_Y \circ f(Q_X\circ F_L(x^*)) \\
		&\leq Q_L \circ F_Y \circ f \{(x/c_X)^{1/\tau_X} + C \log(x) x^{1/\tau_X}/x \},
\end{align*}
where the last relation follows by \eqref{eq:qx:weibull2}. Moreover, appealing to \eqref{eq:fy:weibull2} we found 
\begin{align*} 		
		\M[Y^* | X^* = x^*]  & \leq   Q_L \circ F_Y \left\{ (x/c_X)^{1/\tau_X} + C \log(x) x^{1/\tau_X}/x \right\}^{\alpha_f} \\
		&= c \left\{ (x/c_X)^{1/\tau_X} + C \log(x) x^{1/\tau_X}/x \right\}^{\tau_X}  + O\{\log(x)\}.
\end{align*} 
where the last relation follows as $\alpha_f \tau_\varepsilon = \tau_X$. 
Moreover, since $\tau_X > 1$, the term $\log(x) x^{1/\tau_X}/x \to 0$, as $x \to 0$. Now notice since $\tau_X > 1$, a Taylor development on the last expression yields 
\beao 
		\M[Y^* | X^* = x^*]  	- c_X^{-1} c \, x \leq O\{\log(x)\}.
\eeao 
Similar calculations yield a lower bound and hence, 
\beao 
		|\M[Y^* | X^* = x^*]  	- a x| = O\{ \log(x)\}, \quad x \to \infty,
\eeao 
recalling that $x = ac_X$.
This conlcudes the proof.  
\end{proof}

\subsection{Proof of Example~\ref{examples:gaussian:copula}}

\begin{proof}[Proof Example~\ref{examples:gaussian:copula}]
	It follows directly from Proposition~\ref{prop:weibull:light:tails} with $\tau_X = \tau_\varepsilon = 2$, $\alpha_f = 1$, $c_X = 1/(2\sigma_X^2), c_\varepsilon = 1/(2\sigma_\varepsilon^2)$. It remains to	 compute the value of $a$. Straightforward computations yields
	\[\theta = \{1 + \sigma_\varepsilon^2/(\rho^2\sigma_X^2)\}^{-1} \]
	and therefore
	\[a =  \theta^{\tau_\eps} +  c_{\varepsilon} c_X^{-1}  (1-\theta)^{\tau_\eps}  = \frac{\rho^2\sigma_X^2}{\sigma_\eps^2 + \rho^2\sigma_X^2} .\]
	This concludes the proof of the form of $a$.
	\end{proof}

\subsection{Proof of Example~\ref{example:beta:not:zero}}\label{lem:last:example}

\begin{proof}[Proof Example~\ref{lem:weibull:light:tails}]
The explicit form of the function $f$ implies $\bar F_{f(X)} =  \exp\{-x^{\tau_f}\}/2$, for $x > 0$.
By \citet[][Remarks 2.2 and 3.1]{asmussen:hashorva:laub:taimre:2017}, there exists  $\gamma >0$ such that
\beam \label{eq:bounds:Fy:light:tails:asymp:dep}
	\bar F_Y(x) \sim Cx^{\gamma} \exp\{-x^{\min\{\tau_f,\tau_\eps\} }\}.
\eeam 
Moreover,
\beam \label{eq:qx:weibull}
Q_X \circ F_L(x)  = x^{1/\tau_X}.
\eeam 
By~\eqref{eq:bounds:Fy:light:tails:asymp:dep} and borrowing the calculations in \eqref{eq:fy:weibull2}, we compute an approximation of the transformation of the response distribution to Laplace margins $ Q_L \circ F_Y(x) =  -\log \bar F_Y(x) - \log 2$ by
\beam \label{eq:fy:weibull} 
|Q_L \circ F_Y(x) - x^{\min\{\tau_f,\tau_\eps\} } | = O\{\log(x)\}, \quad x \to \infty.  
\eeam 
Then, 
\begin{align*}
	\M[Y^* | X^* = x^*]  &= Q_L \circ F_Y( \M[Y | X = Q_X\circ F_L(x^*)] ) \\
		&= Q_L \circ F_Y \circ f(Q_X\circ F_L(x^*)) \\
		&= Q_L \circ F_Y ( (x^*)^{1/\tau_f}),
\end{align*}
where the last relation follows by \eqref{eq:qx:weibull}. Moreover, by \eqref{eq:fy:weibull} we find
\begin{align*} 		
		\M[Y^* | X^* = x^*]  & =  Q_L \circ F_Y ( (x^*)^{1/\tau_f} ) \\
		&= (x^*)^{\min\{1,\tau_\eps/\tau_f\} }  + O\{\log(x^*)\}.
\end{align*} 
Finally, 
since $X$, $\varepsilon$, $f(X)$ and $Y$ have exact Weibull-like tails, the assumptions of Theorem~\ref{thm:main} have been verified previously in Section~\ref{subsec:examples:tails}. This concludes the of this example. 
 \end{proof}

\end{document}